\documentclass[a4paper,11pt]{article}
\pdfoutput=0

\usepackage{amssymb,amsmath,latexsym,amsthm}
\usepackage{color,mathrsfs}
\usepackage{url}
\usepackage{enumerate}
\usepackage[numbers]{natbib}
\usepackage{dsfont}
\usepackage{diagbox}

\usepackage{colortbl}

\usepackage{endnotes}

\usepackage[dvips]{graphics}
\usepackage[xetex]{graphicx}
\usepackage{fancyhdr}
\usepackage{pstricks,pst-plot,pst-node,pstricks-add}

\title{2D Theta Functions and Crystallization among Bravais Lattices }

\author{Laurent B\'{e}termin\thanks{laurent.betermin@u-pec.fr}\\ \\ Universit\'e Paris-Est Cr\'eteil \\ LAMA - CNRS UMR 8050 \\ 61, Avenue du G\'en\'eral de Gaulle, \\ 94010 Cr\'eteil, France  }

\newtheorem{thm}{Theorem}[section]
\newtheorem{defi}{Definition}[section]
\newtheorem{corollary}[thm]{Corollary}
\newtheorem{prop}[thm]{Proposition}

\newtheorem{lemma}[thm]{Lemma}

\theoremstyle{definition}
\newtheorem{remark}[thm]{Remark}

\newtheorem{example}[thm]{Example}
\newtheorem{examples}[thm]{Examples}

\newcommand{\R}{\mathbb R}
\newcommand{\Q}{\mathbb Q}
\newcommand{\Z}{\mathbb Z}
\newcommand{\N}{\mathbb N}
\newcommand{\C}{\mathbb C}
\numberwithin{equation}{section}

\def\XXint#1#2#3{{\setbox0=\hbox{$#1{#2#3}{\int}$}
    \vcenter{\hbox{$#2#3$}}\kern-.5\wd0}}

\allowdisplaybreaks

\begin{document}
\maketitle
\begin{abstract} In this paper, we study minimization problems among Bravais lattices for finite energy per point. We prove -- as claimed by Cohn and Kumar -- that if a function is completely monotonic, then the triangular lattice minimizes energy per particle among Bravais lattices with density fixed for any density. Furthermore we give an example of convex decreasing positive potential for which triangular lattice is not a minimizer for some densities. We use the Montgomery method presented in our previous work to prove minimality of triangular lattice among Bravais lattices at high density for some general potentials. Finally, we deduce global minimality among all Bravais lattices, i.e. without a density constraint, of a triangular lattice for some parameters of Lennard-Jones type potentials and attractive-repulsive Yukawa potentials. 
\end{abstract}


\tableofcontents

\section{Introduction and statement of the main results}
The two-dimensional crystallization phenomenon -- that is to say the formation of periodic structures in matter, most of the time at very low temperatures, -- is well known and observed. For instance, similarly to \cite{DustyPlasma}, the following may be mentioned : Langmuir monolayers, Wigner crystal\footnote{In a system of interacting electrons, where the coulomb interaction energy between them sufficiently dominates the kinetic energy or thermal fluctuations}, rare gas atoms adsorbed on graphite, colloidal suspensions, dusty plasma and, from another point of view, vortices in superconductors. In all these cases, particle interactions are complex (quantum effects, kinetic energy, forces related to the environment) and this implies that the physical and mathematical understanding of this kind of problem is highly complicated. However, we would like to know the precise mechanisms that favour the emergence of these periodic structures in order to predict crystal shapes or to build new materials.\\ \\
Semiempirical model potential with experimentally determined parameters are widely used in various physical and chemical problems, and for instance in Monte Carlo simulation studies of clusters and condensed matter. A widespread model is the radial potential, also called ``two-body potential", which corresponds to interaction only depending on distances between particles. This kind of potential, based on approximations, seems to be effective to show the behaviour of matter at very low temperature, when potential energy dominates the others. There are many examples, that can be found in \cite{Kaplan}, but they are usually constructed, except for very simple models such as Hard-sphere, with inverse power laws and exponential functions, easily calculated with a computer if we consider a very large number of particles. For instance we can cite :\begin{itemize}
\item the Lennard-Jones potential $\displaystyle r\mapsto \frac{a_2}{r^{x_2}}-\frac{a_1}{r^{x_1}}$, where the attractive term corresponds to the dispersion dipole-dipole (van der Waals : $\sim r^{-6}$) interaction, initially proposed by Lennard-Jones in \cite{LJ} to study the thermodynamic properties of rare gases and now widely used to study various systems, the best know being for $(x_1,x_2)=(6,12)$;
\item the Buckingham potential $\displaystyle r\mapsto a_1 e^{-\alpha r}-\frac{a_2}{r^6}-\frac{a_3}{r^8}$ proposed by Buckingham in \cite{Buck} and including attractive terms due to the dispersion dipole-dipole ($\sim r^{-6}$) and dipole-quadrupole ($\sim r^{-8}$) interactions, and repulsive terms approximated by an exponential function;
\item the purely repulsive screened Coulomb potential $\displaystyle r\mapsto a\frac{e^{-br}}{r}$, also called ``Yukawa potential", proposed by Bohr in \cite{Bohr} for short atom-atom distances and used for describing interactions in colloidal suspensions, dusty plasmas and Thomas-Fermi model for solids \cite{YBLB,Betermin:2014fy};
\item the Born-Mayer potential $\displaystyle r\mapsto ae^{-br}$ used by Born and Mayer in \cite{BornMayer} in their study of the properties of ionic crystals in order to describe the repulsion of closed shells of ions.
\end{itemize}
Many mathematical works\footnote{We cite only papers about 2D problems.} were conducted with various assumptions on particles interaction : hard sphere potentials \cite{Rad2,Rad3}; oscillating potentials \cite{Suto1}; radial (parametrized or not) potentials \cite{VN3,Crystal,ELi,FriSchmidtYeung}; molecular simulations with radial potentials \cite{Torquato1,TorquatoHCb,YBLB}; three-body (radial and angle parts) potentials \cite{Stef1,Stef2,MainPiovanoStef}; radial potentials and crystallization among Bravais lattices (Number Theory results and applications) \cite{Eno2,Rankin,Cassels,OPS,Gruber,Coulangeon:2010uq,Coulangeon:kx,CoulLazzarini,SarStromb,Mont,Betermin:2014fy}; vortices, in superconductors, among Bravais lattice configurations \cite{Sandier_Serfaty,Zhang,Betermin:2014rr}. Writing these problems in terms of energy minimization is common to all these studies. Furthermore, in many cases, triangular lattices (also called ``Abrikosov lattices" in Ginzburg-Landau theory \cite{Abrikosov}, or sometimes ``hexagonal lattices"), which achieves the best-packing configuration in two dimensions, is a minimizer for the corresponding energy. \\ \\

\begin{center}
\includegraphics[width=10cm,height=80mm]{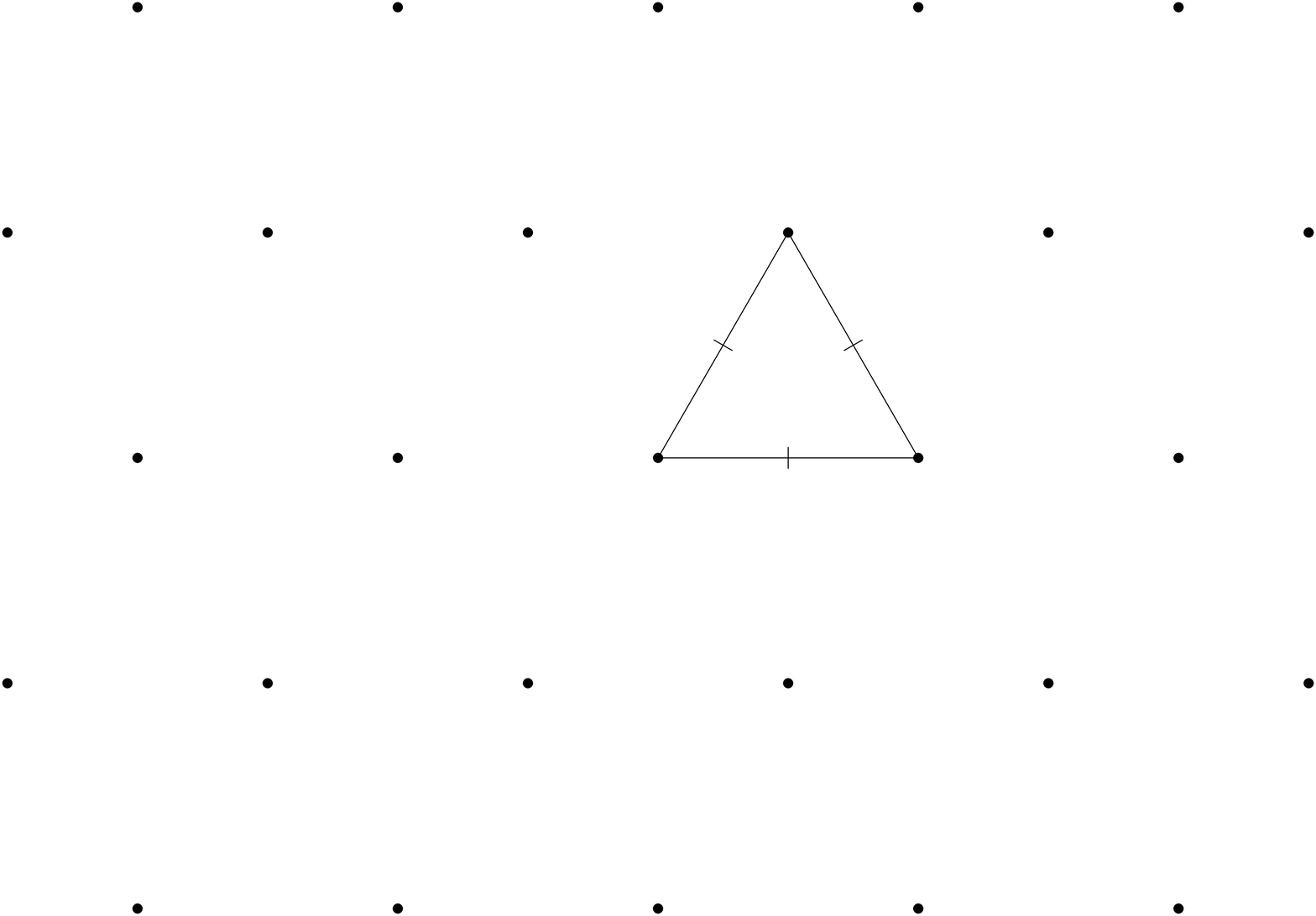}\\
\textbf{Fig. 1 :} Triangular lattice
\end{center}
A clue to understanding this optimality, which is claimed in \cite[p. 139]{CohnKumar}\footnote{A proof of this assertion will be given in Section 3.1.}, is the fact that triangular lattice minimizes, among Bravais lattices, at fixed density, energies
$$
L\mapsto E_{f}[L]:=\sum_{p\in L\backslash\{0\}}f(\|p\|^2)
$$
where $\|.\|$ denotes Euclidean norm in $\R^2$ and $f:\R_+^*\to \R$ is a completely monotonic function, i.e. $\forall k\in\N, \forall r\in \R_+^*, (-1)^k f^{(k)}(r)\geq 0$. Moreover, Cohn and Kumar conjectured, in \cite[Conjecture 9.4]{CohnKumar}, that the triangular lattice seems to minimize energies $E_{f}$ among complex lattices, i.e. union of Bravais lattices, with a fixed density. Hence, it is not surprising, as for the Lennard-Jones potential we studied in \cite{Betermin:2014fy}, that some non-convex sums of completely monotonic functions give triangular minimizer for their energies at high fixed density. We observed this behaviour in works of Torquato et al. \cite{Torquato1,TorquatoHCb}. However, it is important to distinguish mathematical results and physical consistency. Indeed, at very high density i.e. when particles are sufficiently close, kinetic and quantum effects cannot be ignored and our model fails. For instance, Wigner crystal appears if the density is sufficiently low and matter obviously cannot be too condensed. Nevertheless, this kind of result is interesting, whether in Number Theory or in Mathematical Physics and this study of energy among Bravais lattices is the first important step in the search for global ground state, i.e. minimizer among all configurations. For instance, we have recently found in \cite{Betermin:2014rr} a deep connexion between behaviour of vortices in the Ginzburg-Landau theory, and more precisely works of Sandier and Serfaty \cite{Sandier_Serfaty,2DSandier,Serfaty:2013fk}, and optimal logarithmic energy on the unit sphere related to Smale $7$th Problem. Thus the optimality of triangular lattice, among Bravais lattices, for a renormalized energy $W$, which is a kind of Coulomb energy between points in the whole plane, gives important information about optimal asymptotic expansion of spherical logarithmic energy thanks to works by Saff et al. \cite{2400418,Brauchart}. \\ \\
The aim of this paper is to prove this minimality of triangular lattice at high density, with the same strategy as in our previous work \cite{Betermin:2014fy}, that is to say the use of Montgomery result \cite{Mont} about optimality of triangular lattice at a fixed density for theta functions
$$
L\mapsto \theta_L(\alpha):=\sum_{p\in L\backslash\{0\}} e^{-2\pi \alpha \|p\|^2},
$$
for some general admissible\footnote{A rigorous definition will be given in preliminaries.} potentials $f$, summable on lattices and such that their inverse Laplace transform $\mu_f$ exists on $[0,+\infty)$. Hence, as in the classical ``Riemann's trick" that we used in \cite{Betermin:2014fy}, we can write an integral representation of energy $E_f$ which we deduce a sufficient condition for minimality of triangular lattice among Bravais lattices of fixed density\footnote{Actually, as in \cite{Betermin:2014fy}, we will write all our results in terms of area, that is to say the inverse of the density.}. This is precisely the aim of our first main theorem,  which we now state.
\vspace{7mm}

\begin{thm}\label{THM1} For any admissible potential $f$, for any $A>0$ and any Bravais lattice $L$ of area $A$, there exists a constant $C_A$, which not depends on $L$, such that
\begin{equation}\label{Intrepresent}
E_f[L]=\frac{\pi}{A}\int_1^{+\infty}\left[\theta_L\left(\frac{y}{2A}  \right)-1  \right]\left[y^{-1}\mu_f\left( \frac{\pi}{yA}  \right)+\mu_f\left( \frac{\pi y}{A} \right)  \right]dy + C_A
\end{equation}
where $\mu_f$ is the inverse Laplace transform of $f$. Moreover, if 
\begin{equation}\label{SuffiCond}
y^{-1}\mu_f\left( \frac{\pi}{yA}  \right)+\mu_f\left( \frac{\pi y}{A} \right) \geq 0 \quad \text{ a.e. on } [1,+\infty)
\end{equation}
then the triangular lattice of area $A$, i.e. $\displaystyle\Lambda_A=\sqrt{\frac{2A}{\sqrt{3}}}\left[\Z (1,0) \oplus \Z (1/2,\sqrt{3}/2)\right]$, is the unique minimizer of $L\mapsto E_f[L]$, up to rotation, among Bravais lattices of fixed area $A$.
\end{thm}
\vspace{7mm}

\noindent Sufficient condition \eqref{SuffiCond} can be applied for some general functions $f$. More precisely we will consider the following potentials\footnote{It is important to distinguish potential $f$ and the function $r\mapsto f(r^2)$ that we sum on lattices to compute $E_{f}$.}, defined for $r>0$, which we will explain the interest throughout the paper : 
\begin{itemize}
\item Sums of screened coulombian potentials : 
$$
 \varphi_{a,x}(r)= \sum_{i=1}^n a_i \frac{e^{-x_i r}}{r},
$$
 with $0<x_1<x_2...<x_n$, $a_i\in \R^*$ for all $1\leq i\leq n$ and $\displaystyle \sum_{i=1}^n a_i\geq 0$;
\item Sums of inverse power laws : 
$$
V_{a,x}(r)=\sum_{i=1}^n \frac{a_i}{r^{x_i}},
$$
with $1<x_1<x_2<...<x_n$, $a_i\in \R^*$ for all $1\leq i\leq n$ and $a_n>0$; 

\item Potentials with exponential decay : 
$$
 f_{a,x,b,t}(r)=V_{a,x}(r)+\sum_{i=1}^m b_i e^{-t_i\sqrt{r}},
 $$
 with $3/2<x_1<x_2<...<x_n$, $a_i\in\R_+^*$ for all $1\leq i\leq n$, $a_n>0$, $b_j,t_j\in\R^*$ for all $1\leq j\leq m$.
\end{itemize}
Thus, even though our method is not optimal, we will give explicit area bounds in Propositions \ref{SCGen}, \ref{HDInvPower}, \ref{LJones} and \ref{Genexpo}, with respect to potential parameters, above which minimizer is triangular and we give conditions on parameters, for potentials $\varphi_{a,x}^{AR}$ and $V_{a,x}^{LJ}$ in order to get a triangular global minimizer, i.e. without area constraint, in particular when the potential has a well. This is the aim of our second theorem.
\vspace{7mm}

\begin{thm}\label{THM2} Let functions $\varphi_{a,x},\varphi_{a,x}^{AR}, V_{a,x}, V_{a,x}^{LJ}$ and $f_{a,x,b,t}$ be defined as before.\\ \\ 
\noindent \textnormal{\textbf{A. Minimality at high density.}} If $f\in \{\varphi_{a,x}, V_{a,x},f_{a,x,b,t}\}$ then there exists $A_0>0$ such that for any $0<A\leq A_0$, $\Lambda_A$ is the unique minimizer, up to rotation, of $L\mapsto E_f[L]$ among Bravais lattices of fixed area $A$. \\ \\
\noindent \textnormal{\textbf{B. Global optimality without an area constraint.}} We have the following two cases
\begin{enumerate}
\item  Let $\varphi_{a,x}^{AR}$ be the attractive-repulsive potential defined by 
$$
 \varphi_{a,x}^{AR}(r)=a_2\frac{e^{-x_2 r}}{r}-a_1\frac{e^{-x_1 r}}{r},
 $$
  where $0<a_1<a_2$ and $0<x_1<x_2$. If $a_1,a_2,x_1,x_2$ satisfy 
\begin{equation}\label{CondMinGlob}
\frac{a_1\left(1+\frac{x_1}{x_2}\pi  \right)}{a_2(1+\pi)}e^{\left( 1-\frac{x_1}{x_2} \right)\pi}\geq 1 \quad \text{ and } \quad \frac{a_1\left( a_1x_2+x_1(a_2-a_1)\pi \right)}{a_2x_2\left(a_1+(a_2-a_1)\pi\right)}e^{\left(1-\frac{x_1}{x_2}  \right)\left( \frac{a_2}{a_1}-1 \right)\pi}\geq 1,
\end{equation}
then the minimizer of $L\mapsto E_{\varphi_{a,x}^{AR}}[L]$ among all Bravais lattices is unique, up to rotation, and triangular. In particular it is true if $a_2=2a_1$ and $x_1\leq 0.695 x_2$.
\item Let $V_{a,x}^{LJ}$ be the Lennard-Jones type potential defined by
$$
V_{a,x}^{LJ}(r)=\frac{a_2}{r^{x_2}}-\frac{a_1}{r^{x_1}},
$$
with $1<x_1<x_2$ and $(a_1,a_2)\in (0,+\infty)^2$. We set $h(t)=\pi^{-t}\Gamma(t) t$. If $h(x_2)\leq h(x_1)$ then the minimizer $L_{a,x}$ of $L\mapsto E_{V_{a,x}^{LJ}}[L]$ among all Bravais lattices is unique, up to rotation, and triangular. Moreover its area is
$$
|L_{a,x}|=\left(\frac{a_2x_2\zeta_{\Lambda_1}(2x_2)}{a_1x_1\zeta_{\Lambda_1}(2x_1)}  \right)^{\frac{1}{x_2-x_1}}.
$$
In particular, it is true if $(x_1,x_2)\in \left\{(1.5,2);(1.5,2.5);(1.5,3);(2,2.5);(2;3) \right\}$\footnote{See Section 6.4 for numerical values.}.
\end{enumerate}
\end{thm}
\vspace{7mm}

\noindent We proceed as follows, we start below with some preliminaries where we recall Montgomery result about optimality of $\Lambda_A$ for theta functions $\theta_L$ and we give the definition of admissible potential. Then we prove in Section 3 the optimality of $\Lambda_A$ for every $A$ when $f$ is completely monotonic and we give an example of strictly convex, decreasing and positive potential $V$ such that $\Lambda_A$ is not a minimizer of $E_f$ for some $A$. Theorem \ref{THM1} is proved in Section 4, with some general applications. Furthermore we discuss optimality and improvement of this method. Finally we prove our Theorem \ref{THM2} in next sections where we present the interest, in molecular simulation, of studied potentials and we prove additional results. Throughout the paper, we give numerical results and examples.

\section{Preliminaries}
\subsection{Bravais lattices, zeta functions and theta functions}
We briefly recall our notations in \cite{Betermin:2014fy}. Throughout this paper, $\|.\|$ will denote the Euclidean norm in $\R^2$. Let $L=\Z u\oplus\Z v$ be a \textbf{Bravais lattice} of $\R^2$, then by Engel's theorem (see \cite{Engel}), we can choose $u$ and $v$ such that $\|u\|\leq \|v\|$ and $\displaystyle (\widehat{u,v})\in\left[\frac{\pi}{3},\frac{\pi}{2}\right]$ in order to obtain the unicity of the lattice, up to rotations and translations and the fact that the lattice is parametrized by its both first lengths $\|u\|$ and $\|v\|$. We note $|L|=\|u\land v\|=\|u\|\|v\|\left|\sin(\widehat{u,v})\right|$ the area\footnote{We choose, as in \cite{Betermin:2014fy}, to write results in terms of area and not in terms of density (which is its inverse).} of $L$ which is in fact the area of its primitive cell. Let $\displaystyle\Lambda_A=\sqrt{\frac{2A}{\sqrt{3}}}\left[\Z (1,0) \oplus \Z (1/2,\sqrt{3}/2)\right]$ be the triangular lattice of area $A$, then $\|u\|$ is called the length of this lattice.\\ \\
For real $s>2$, the \textbf{Epstein zeta function} of a Bravais lattice $L$ is defined by
$$
\zeta_L(s)=\sum_{p\in L^*}\frac{1}{\|p\|^s}.
$$
where $L^*:=L\backslash\{0\}$. As proved in \cite[Proposition 10.5.5 and Proposition 10.5.7]{Cohen}, we can write $\zeta_L(s)$ in term of $L$-function or Hurwitz zeta-function. More precisely, for $L=\Z^2$ and $L=\Lambda_1$ the triangular lattice of area $1$, we have, for any $s>1$,
\begin{align}
\label{zetasquare}&\zeta_{\Z^2}(2s)=4L_{-4}(s)\zeta(s)=4^{-s+1}\zeta(s)\left[ \zeta(s,1/4)-\zeta(s,3/4) \right],\\
\label{zetatriang}&\zeta_{\Lambda_1}(2s)=6\left(\frac{\sqrt{3}}{2}  \right)^s \zeta(s)L_{-3}(s)=6\left(\frac{\sqrt{3}}{2}  \right)^s 3^{-s}\zeta(s)\left[\zeta(s,1/3)-\zeta(s,2/3)  \right],
\end{align}
where $\zeta$ is the classical Riemann zeta function $\displaystyle \zeta(s):=\sum_{i=1}^{+\infty}n^{-s}$, $L_D$ defined by
$$
L_D(s):=\sum_{n=1}^{+\infty}\left( \frac{D}{n} \right)n^{-s}
$$
is the Dirichlet $L$-function associated to quadratic field $\Q(i\sqrt{-D})$, with $\displaystyle \left( \frac{D}{n} \right)$ the Legendre symbol, and, for $x>0$,
$$
\zeta(s,x):=\sum_{n=0}^{+\infty} (n+x)^{-s}
$$
is the Hurwitz zeta function. Hence both these special values are easily computable. \\ \\
Now we recall fundamental Montgomery's Theorem about optimality of $\Lambda_A$ among Bravais lattices for theta functions :
\begin{thm}\label{Mgt} \textnormal{(Montgomery, \cite{Mont})} For any real number $\alpha>0$ and a Bravais lattice $L$, let
\begin{equation} \label{Thetadef}
\theta_L(\alpha):=\Theta_L(i\alpha)=\sum_{p\in L}e^{-2\pi\alpha\|p\|^2},
\end{equation}
where $\Theta_L$ is the Jacobi \textbf{theta function} of the lattice $L$ defined for $Im(z)>0$. Then, for any $\alpha>0$, $\Lambda_A$ is the unique minimizer of $L\to \theta_L(\alpha)$, up to rotation, among Bravais lattices of area $A$.
\end{thm}
\begin{remark} This result implies that the triangular lattice is the unique minimizer, up to rotation, of $L\mapsto\zeta_L(s)$ among Bravais lattices with density fixed for any $s>2$ which is also proved by Rankin in \cite{Rankin}, Cassels in \cite{Cassels}, Ennola in \cite{Ennola} and Diananda in \cite{Diananda}. Montgomery deduced this fact by the famous ``Riemann's trick" (see \cite{Terras} or \cite{Betermin:2014fy} for a proof): for any $L$ such that $D=1$,
\begin{equation} \label{Riemann}
\text{for } \textnormal{Re}(s)>1,\quad  \zeta_L(2s)\Gamma(s)(2\pi)^{-s}=\frac{1}{s-1}-\frac{1}{s}+\int_1^\infty(\theta_L(\alpha)-1)(\alpha^s+\alpha^{1-s})\frac{d\alpha}{\alpha}.
\end{equation}
In Section 4.2 we will prove general Riemann's trick \ref{Intrepresent}, which we call integral representation of energy, for admissible potentials in order to use Montgomery method in a general case.
\end{remark} 

\begin{remark}
We can find in \cite[Appendice A]{NonnenVoros} other proof of minimality of some theta functions based on result of Osgood, Phillips and Sarnak \cite[Corollary 1(b) and Section 4]{OPS} about Laplacian's determinant of flat torus, which has some deep connection with other energies (for instance, see \cite[Theorem 2.3]{Betermin:2014rr}).
\end{remark}

\subsection{Admissible potential, inverse Laplace transform and lattice energies}
\begin{defi}\label{Admis}
We say that $f:\{Re(z)>0\}\to \R$ is \textbf{admissible} if :
\begin{enumerate}
\item there exists $\eta>1$ such that $|f(z)|=O(|z|^{-\eta})$ as $|z|\to +\infty$;
\item $f$ is analytic on $\{z\in \C; Re(z)>0 \}$.
\end{enumerate}
If $f$ is admissible, we define, for any Bravais lattice $L$ of $\R^2$,
$$
E_f[L]:=\sum_{p\in L^*}f(\|p\|^2)
$$
which is \textbf{the quadratic energy per point of lattice $L$ created by potential $f$}.

\end{defi}
\begin{remark}
As a consequence of \cite[Theorem 5.17, Theorem 5.18]{Poularikas}, we get, by direct application of inversion integral formula :
\begin{itemize}
\item There exists an unique \textbf{inverse Laplace transform} $\mu_f$\footnote{We will sometimes write $\mathcal{L}$ and $\mathcal{L}^{-1}$ for Laplace and inverse Laplace operators.}, which is continuous on $(0,+\infty)$;
\item We have $\mu_f(0)=0$.
\end{itemize}
\end{remark}

\begin{remark}
This definition excludes two-dimensional Coulomb potential $r\mapsto -\log r$ because all its quadratic energies are infinite. However we can define a renormalized energy as in \cite{Sandier_Serfaty} or in \cite{SaffLongRange}.
\end{remark}

\subsection{Completely monotonic functions}
The class of completely monotonic functions is central in our work. Indeed, as we will see in Sec. 3, these functions have good properties for our problem of minimization among lattices with fixed area thanks to the Montgomery theorem \ref{Mgt}.

\begin{defi}
A $C^\infty$ function $f:(0,+\infty)\to \R_+$ is said to be \textbf{completely monotonic} if, for any $k\in\N$ and any $r>0$,
$$
(-1)^kf^{(k)}(r)\geq 0.
$$
\end{defi}
\begin{examples}We can find a lot of examples of completely monotonic functions in \cite{ComplMonotonic}. Here we give only some interesting classical admissible potentials $f$ :
\begin{itemize}
\item $V_x(r)=r^{-x},x>1$;
\item $\displaystyle V_{a,x}(r)=\sum_{i=1}^n a_i r^{-x_i}$ where $a_i>0$ and $x_i>1$ for all $i$;
\item $f_\alpha(r)=e^{-a r^{\alpha}},a>0,\alpha\in (0,1]$, see \cite[Corollary 1]{ComplMonotonic};
\item Modified Bessel function, i.e. one of the two solutions of $r^2 y''+r y'-(r^2+\nu^2)y=0 
$ which goes to $0$ at infinity, is $K_\nu(r)=\int_0^{+\infty}e^{-r\cosh t}\cosh(\nu t)dt$, $\nu\in \R$. Moreover, $r\mapsto K_\nu(\sqrt{r})$ is also completely monotonic (despite we thought in \cite{Betermin:2014fy}).
\item $\displaystyle V_{SC}(r)=\frac{e^{-a\sqrt{r}}}{\sqrt{r}}$, $a>0$;
\item $\displaystyle \varphi_a(r)=\frac{e^{-ar}}{r}$, $a>0$.
\end{itemize}
\end{examples}

\begin{remark}
We remark that if $r\mapsto f(r)$ is completely monotonic, it is not generally the case for $r\mapsto f(r^2)$. For instance $r\mapsto e^{-r}$ is completely monotonic, but $r\mapsto e^{-r^2}$ does not check this property.
\end{remark}
\noindent Now we give the famous connection between completely monotonic function and Laplace transform due to Bernstein in \cite{Bernstein}.

\begin{thm}\label{Bernstein} \textbf{(Hausdorff-Bernstein-Widder Theorem)} A function $f:\R_+^*\to\R$ is completely monotonic on $\R_+$ if and only if it is the Laplace transform of a finite non-negative Borel measure $\mu$ on $\R_+$, i.e.
$$
f(r)=\mathcal{L}[\mu](r)=\int_0^{+\infty}e^{-rt}d\mu(t).
$$
\end{thm}

\begin{remark} If $f$ is admissible and completely monotonic, then 
$$
d\mu(t)=\mu_f(t)dt \quad \text{ and }\quad \mu_f(t)\geq 0, \text{ a.e. on }(0,+\infty).
$$
\end{remark}
\begin{remark}\label{remarkBochner}
Actually Schoenberg proved in \cite{Schoenberg} that $r\mapsto f(r)$ is completely monotonic if and only if $r\mapsto f(r^2)$ is a positive definite function in $\R$, i.e. for any $N\in \N\backslash\{0,1\}$, any $x_1,...,x_N \in \R$ and any $c_1,...,c_N \in \R$, we have
$$
\sum_{i,j=1}^N c_i c_j f(|x_i-x_j|^2)\geq 0
$$
or, by Bochner Theorem (see \cite{Bochner}), if and only if $r\mapsto f(r^2)$ is the Fourier transform of a positive finite Borel measure on $\R$. \\ \\
Positivity of Fourier transform of a radial potential is a key point in crystallization problems. Indeed Nijboer and Ventevogel proved in \cite{VN3} that it is a necessary condition for a periodic ground state (Bravais lattices) and S\"{u}to studied in his work \cite{Suto1} potentials $f$ such that $\hat{f}(k)\geq 0$ and $\hat{f}(k)=0$ for any $\|k\|>R_0$ and proved some interesting crystallization results at high densities. Unfortunately, as Likos explained in \cite{Likos}, this kind of potential, oscillating and with inverse power law decay, seems to be difficult to achieve physically.\\ \\
Actually it is more common to use Fourier transform in problems of minimization of lattice energy because we have the Poisson summation formula and the natural periodicity of sinus and cosinus. Furthermore, applications of classical formula allows to obtain some interesting results, as in \cite[Proposition 9.3]{CohnKumar}. However we will show in Section 4 that inverse Laplace transform also seems well adapted to our problem and gives simple calculations. Indeed, Fourier methods as in \cite{CohnKumar,Suto1,Suto2} is good for more general minimization problems and our method is a better choice for minimization among Bravais lattices because of integral representation \eqref{Intrepresent}.
\end{remark}

\subsection{Cauchy's bound for positive root of a polynomial}
In this part, we recall Cauchy's rule explained in \cite[Note III, Scolie 3, page 388]{Cauchy} for upper bound of polynomial's positive roots (see also \cite{Vigklas} for simple proof).
\begin{thm} \label{Cauchy} \textbf{(Cauchy's rule)}
Let $P$ a polynomial of degree $n>0$ defined by
$$
P(X)=\sum_{i=0}^n\alpha_i X^i, \quad \alpha_n>0
$$
where $\alpha_i<0$ for at least one $i$, $0\leq i\leq n-1$. If $\lambda$ is the number of negative coefficients, then an upper bound on the values of the positive roots of $P$ is given by
$$
M_P=\max_{i;\alpha_i<0} \left\{\left(\frac{-\lambda \alpha_i}{\alpha_n} \right)^{\frac{1}{n-i}}\right\}
$$
\end{thm}

\begin{remark}
This Theorem stays true for upper bound on the values of the positive zero of any function $p$ defined by
$$
p(y)=\sum_{i=1}^n\alpha_i y^{\nu_i},\quad \alpha_n>0
$$
where $0<\nu_1<...<\nu_n$ are real numbers and we obtain
\begin{equation}\label{UBoundRoot}
M_p=\max_{i;\alpha_i<0} \left\{\left(\frac{-\lambda \alpha_i}{\alpha_n} \right)^{\frac{1}{\nu_n-\nu_i}}\right\}.
\end{equation}
This result will be useful for technical reasons in the following sections, because we will want positive zeros less than $1$ to apply our sufficient condition inTheorem \ref{THM1} and to prove Theorem \ref{THM2}.A.
\end{remark}

\section{Completely monotonic functions and optimality of $\Lambda_A$}
In this part we begin to state a simple fact connecting positivity of inverse Laplace transform and minimality among lattices at fixed area. Furthermore we will give an example of strictly convex, decreasing, positive potential for which there exists areas so that the triangular lattice is not a minimizer among Bravais lattices with fixed area.

\subsection{Optimality at any density}
The following proposition, claimed by Cohn and Kumar in \cite[page 139]{CohnKumar}, is a natural consequence of Montgomery and Hausdorff-Bernstein-Widder Theorems.
\begin{prop} \label{CrystCM} Let $f$ be an admissible potential. If $f$ is completely monotonic then, for any $A>0$, $\Lambda_A$ is the unique minimizer, up to rotation, of 
$$
L \mapsto E_f[L]=\sum_{p\in L^*}f(\|p\|^2)
$$
among lattices of fixed area $A$.
\end{prop}
\begin{proof}
As $f$ is admissible, we can write,
$$
f(r)=\int_0^{+\infty} e^{-rt}\mu_f(t)dt
$$
and it follows that
\begin{align*}
E_f[L]&=\sum_{p\in L^*}f(\|p\|^2)=\sum_{p\in L^*}\int_0^{+\infty}e^{-t\|p\|^2}\mu_f(t)dt=\int_0^{+\infty}\sum_{p\in L^*}e^{-t\|p\|^2}\mu_f(t)dt\\
&=\int_0^{+\infty}\left[\theta_L\left(\frac{t}{2\pi}\right)-1\right]\mu_f(t)dt
\end{align*}
and
$$
E_f[L]-E_f[\Lambda_A]=\int_0^{+\infty}\left[\theta_L\left(\frac{t}{2\pi}\right)-\theta_{\Lambda_A}\left(\frac{t}{2\pi}\right)\right]\mu_f(t)dt.
$$
If $f$ is completely monotonic, by Theorem \ref{Bernstein}, $\mu_f(r)\geq 0$ for almost every $r\in (0,+\infty)$. Moreover, by Montgomery Theorem \ref{Mgt}, for any $t> 0$ and any Bravais lattice $L$ of area $A$, 
$$
\theta_L\left(\frac{t}{2\pi}\right)-\theta_{\Lambda_A}\left(\frac{t}{2\pi}\right)\geq 0,
$$
hence $E_f[L]\geq E_f[\Lambda_A]$ for any $L$ such that $|L|=A$ and $\Lambda_A$ is the unique minimizer of the energy among Bravais lattices of fixed area $A$.
\end{proof}

\begin{remark}
We can imagine that the reciprocal is true, i.e. if $f$ is not completely monotonic, then there exists $A_0$ such that $\Lambda_{A_0}$ is not a minimizer among Bravais lattices of area $A_0$ fixed. In next subsection will give an explicit example correlated with Marcotte, Stillinger and Torquato results in \cite{Torquato1} about the existence of unusual ground states with convex decreasing positive potential.
\end{remark}
\begin{examples} A direct consequence of this theorem is the minimality of triangular lattice among lattices for any fixed area for the following energies :
\begin{itemize}
\item $E_{V_x}[L]=\zeta_L(2x)$, $x>1$ is the first natural example given by Montgomery in \cite{Mont},
\item $\displaystyle E_{V_{a,x}}[L]=\sum_{i=1}^n a_i\zeta_L(2x_i)$ where $a_i>0$ and $x_i>1$ for all $i$,  
\item $\displaystyle E_{f_\alpha}[L]=\sum_{p\in L^*}e^{-a\|p\|^{2\alpha}}, \alpha\in (0,1]$, in particular $\displaystyle E_{f_{1/2}}[L]=\sum_{p\in L^*}e^{-a\|p\|}$,
\item $\displaystyle E_{K_\nu(\sqrt{.})}[L]=\sum_{p\in L^*} K_\nu(\|p\|)$ , $\nu \in \R$ which generalizes our study of lattice energy with potential $K_0$ in \cite{Betermin:2014fy} in Thomas-Fermi model case;
\item $\displaystyle E_{V_{SC}}[L]=\sum_{p\in L^*} \frac{e^{-a\|p\|}}{\|p\|}$, $a>0$, which corresponds to lattice energy for screened Coulomb potential interaction and can explain formation of triangular Wigner crystal at low density \cite{GrimesAdams};
\item $\displaystyle E_{\varphi_a}[L]=\sum_{p\in L^*} \frac{e^{-a\|p\|^2}}{\|p\|^2}$, $a>0$.
\end{itemize}
\end{examples}

\subsection{Repulsive potential and triangular lattice}
\noindent In this section we give an example of stricly convex decreasing positive radial potential $V$ so that, for some densities, a minimizer of $E_V$ among Bravais lattices of density fixed cannot be triangular. As Ventevogel and Nijboer proved in \cite{VN1}, a convex decreasing positive potential allows to obtain, in one dimension and for any fixed density, a dilated of lattice $\Z$ as unique minimizer among all configurations. Thus the two-dimensional case is deeply different.\\ \\
Let 
\begin{equation}\label{CE}
V(r)=\frac{14}{r^2}-\frac{40}{r^3}+\frac{35}{r^4}
\end{equation}
be the potential and 
$$
E_V[L]=14\zeta_L(4)-40\zeta_L(6)+35\zeta_L(8)
$$
the energy per point of a Bravais lattice $L$. 
\begin{center}
\includegraphics[width=10cm,height=80mm]{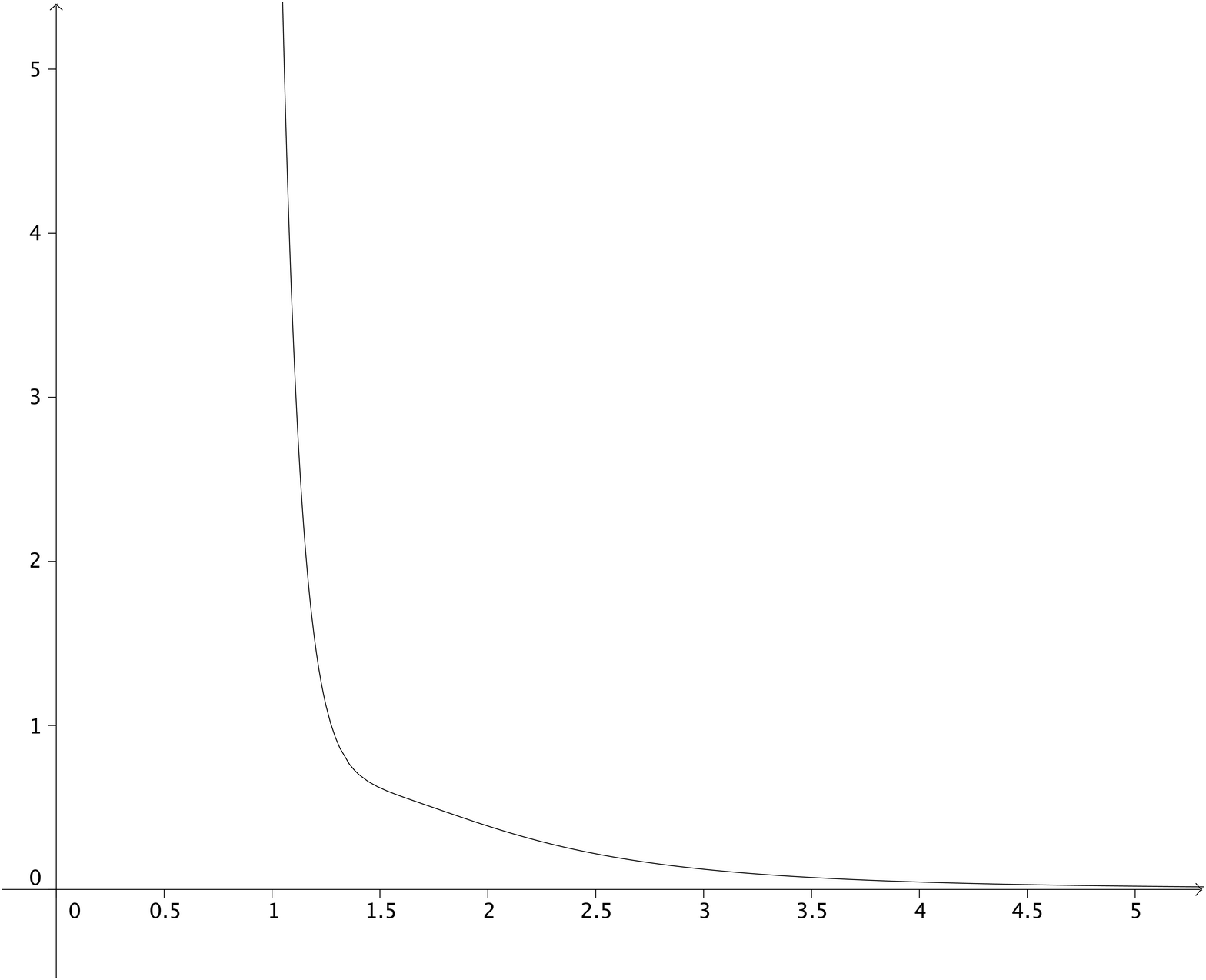}\\
\textbf{Fig. 2 :} Graph of $r\mapsto V(r^2)$
\end{center}

\begin{prop} \label{nocryst} \textbf{(Strictly convex potential and non optimality of triangular lattice)} 
Let $V$ be given by $\eqref{CE}$, then
\begin{itemize}
\item $V$ is strictly positive, strictly decreasing and strictly convex on $(0,+\infty)$;
\item There exists $A_1,A_2$ such that \textbf{$\Lambda_A$ is not a minimizer} of $E_V$ among all Bravais lattices of area $A\in (A_1,A_2)$.
\end{itemize}
\end{prop}
\begin{proof}
We have
$$
V(r)=\frac{14r^2-40r+35}{r^4}
$$
and the discriminant of polynomial $14X^2-40X+35$ is $\Delta_1=-360<0$, hence $V(r)>0$ on $(0,+\infty)$.\\
We compute
$$
V'(r)=\frac{-4(7r^2-30r+35)}{r^5}
$$
and the discriminant of $7X^2-30X+35$ is $\Delta_2=-80<0$, therefore $V'(r)<0$, i.e. $V$ is strictly decreasing on $(0,+\infty)$.\\
Moreover, we have
$$
V''(r)=\frac{4(21r^2-120r+175)}{r^6}
$$
and the discriminant of $21X^2-120X+175$
 is $\Delta_3=-300<0$, then $V''(r)>0$ on $(0,+\infty)$, i.e. $V$ is strictly convex on $(0,+\infty)$.\\
For the second point, we have the following equivalences
\begin{align*}
& E_V[L]\geq E_V[\Lambda_A] \text{ for any } |L|=A\\
&\iff 14\zeta_L(4)-40\zeta_L(6)+35\zeta_L(8) \geq 14\zeta_{\Lambda_A}(4)-40\zeta_{\Lambda_A}(6)+35\zeta_{\Lambda_A}(8)\geq 0  \text{ for any } |L|=A\\
& \iff \frac{14}{A^2}\left(\zeta_L(4)-\zeta_{\Lambda_1}(4)   \right)+\frac{40}{A^3}\left(\zeta_L(6)-\zeta_{\Lambda_1}(6)   \right)+\frac{35}{A^4}\left(\zeta_L(8)-\zeta_{\Lambda_1}(8)   \right)\geq 0 \text{ for any } |L|=1\\
& \iff 14\left(\zeta_L(4)-\zeta_{\Lambda_1}(4)   \right)A^2-40\left(\zeta_L(6)-\zeta_{\Lambda_1}(6)   \right)A+35\left(\zeta_L(8)-\zeta_{\Lambda_1}(8)   \right)\geq 0 \text{ for any } |L|=1\\
&\iff P_{L}(A)\geq 0 \text{ for any } |L|=1
\end{align*}
where the discriminant of $P_L(A)=14\left(\zeta_L(4)-\zeta_{\Lambda_1}(4)   \right)A^2-40\left(\zeta_L(6)-\zeta_{\Lambda_1}(6)   \right)A+35\left(\zeta_L(8)-\zeta_{\Lambda_1}(8)   \right)$ is
$$
\Delta(L)=1600\left(\zeta_L(6)-\zeta_{\Lambda_1}(6)   \right)^2-1960\left(\zeta_L(4)-\zeta_{\Lambda_1}(4)   \right)\left(\zeta_L(8)-\zeta_{\Lambda_1}(8)   \right).
$$
For $L=\Z^2$ the square lattice of area $1$, we obtain $\Delta(\Z^2)\approx 24.231435>0$ then there exist two positive numbers $A_1$ and $A_2$ such that $P_{\Z^2}(A)<0$ for any $A_1<A<A_2$. Hence, $\Lambda_A$ is not a minimizer of $E_V$ among Bravais lattices with fixed area $A$ if $A_1<A<A_2$. More precisely we get
$$
A_1\approx 2.3152307 \text{ and } A_2\approx 3.759353.
$$
\end{proof}
 \begin{remark}
 It follows, from the previous proof, that function $r\mapsto V(r^2)$ is also strictly positive, strictly decreasing and strictly convex on $(0,+\infty)$.
 \end{remark}
\begin{remark}
Actually, the previous proof implies that, for any $A\in (A_1,A_2)$,
$$
E_V[\sqrt{A}\Z^2]<E_V[\Lambda_A].
$$
Moreover, this interval seems numerically to be optimal, i.e. for any $A\not\in [A_1,A_2]$, $\Lambda_A$ seems to be the unique minimizer, up to rotation, of $L\mapsto E_f[L]$ among Bravais lattices of fixed area $A$.
\end{remark}

\section{Sufficient condition and first applications}
Now we study the case of non completely monotonic potential $f$, i.e. $\mu_f$ is negative on a subset of $(0,+\infty)$ of positive Lebesgue measure.

\subsection{Integral representation and sufficient condition : Proof of Theorem \ref{THM1} }

\begin{proof}
Let $L$ be a Bravais lattice of area $A$ and $f$ an admissible potential. Firstly we prove the integral representation \ref{Intrepresent} of energy $E_f[L]$ :
\begin{align*}
E_f[L]:=\sum_{p\in L^*}f(\|p\|^2)=&\frac{\pi}{A}\int_1^{+\infty}\left[\theta_L\left(\frac{y}{2A}  \right)-1  \right]\left(y^{-1}\mu_f\left(\frac{\pi}{yA}\right)+\mu_f\left(\frac{\pi y}{A}  \right)   \right)dy\\
&+\frac{\pi}{A}\int_1^{+\infty}\mu_f\left(\frac{\pi}{yA}  \right)(y^{-1}-y^{-2})dy.
\end{align*}
Indeed, for a Bravais lattice $L$ of $\R^2$ with $|L|=1/2$, we have, as in \cite{Betermin:2014fy}, by $t=2\pi u$, $u=y^{-1}$ and Montgomery's identity $\theta_L(y^{-1})=y\theta_L(y)$ (proved in \cite{Mont}) :
\begin{align*}
&E_f[L]:=\sum_{p\in L^*}f(\|p\|^2)=\sum_{p\in L^*}\int_0^{+\infty} e^{-t\|p\|^2}\mu_f(t)dt=2\pi\sum_{p\in L^*} \int_0^{+\infty}e^{-2\pi u\|p\|^2}\mu_f(2\pi u)dt\\
&=2\pi\int_0^{+\infty} \left[\theta_L(u)-1  \right]\mu_f(2\pi u)du\\
&=2\pi\int_0^{1} \left[\theta_L(u)-1  \right]\mu_f(2\pi u)du+2\pi\int_1^{+\infty} \left[\theta_L(u)-1  \right]\mu_f(2\pi u)du\\
&= 2\pi\int_1^{+\infty}\left[\theta_L(y^{-1})-1\right]\mu_f\left( \frac{2\pi}{y} \right)\frac{dy}{y^2}+2\pi\int_1^{+\infty} \left[\theta_L(u)-1  \right]\mu_f(2\pi u)du\\
&=2\pi\int_1^{+\infty}\left[y\theta_L(y)-1\right]\mu_f\left( \frac{2\pi}{y} \right)\frac{dy}{y^2}+2\pi\int_1^{+\infty} \left[\theta_L(u)-1  \right]\mu_f(2\pi u)du\\
&=2\pi\int_1^{+\infty}\theta_L(y) \mu_f\left( \frac{2\pi}{y} \right)\frac{dy}{y}+2\pi\int_1^{+\infty} \left[\theta_L(u)-1  \right]\mu_f(2\pi u)du-2\pi\int_1^{+\infty}\mu_f\left( \frac{2\pi}{y} \right)\frac{dy}{y^2}\\
&=2\pi\int_1^{+\infty}[\theta_L(y)-1]\left(y^{-1}\mu_f\left( \frac{2\pi}{y} \right) +\mu_f(2\pi y) \right)dy+2\pi\int_1^{+\infty}\mu_f\left(\frac{2\pi}{y}  \right)(y^{-1}-y^{-2})dy
\end{align*}
Now we have, by change of variable $t=y^{-1}$,
\begin{align*}
\left|\int_1^{+\infty}\mu_f\left(\frac{2\pi}{y}  \right)(y^{-1}-y^{-2})dy\right|&\leq \int_0^1\left|\mu_f(2\pi t)\right|(t^{-1}-1)dt<+\infty
\end{align*}
because $\mu_f$ is continuous on $\R_+^*$, $\mu_f(0)=0$ and $t\mapsto t^{-1}$ is integrable at the neighbourhood of $0$.\\
Hence, for $L$ such that $|L|=A$, we have 
$$
E_f[L]=\sum_{p\in L^*}f(\|p\|^2)=\sum_{p\in \tilde{L}^*}f(2A\|p\|^2)
$$
where $L=\sqrt{2A}\tilde{L}$, $|\tilde{L}|=1/2$. By identities $\displaystyle \mu_{f(k.)}=\frac{1}{k}\mu_f\left( \frac{.}{k} \right)$ and $\displaystyle \theta_{\tilde{L}}(s)=\theta_L\left(  \frac{s}{2A}\right)$, we get
\begin{align*}
E_f[L]=&\frac{\pi}{A}\int_1^{+\infty}\left[\theta_L\left(\frac{y}{2A}  \right)-1  \right]\left(y^{-1}\mu_f\left(\frac{\pi}{yA}\right)+\mu_f\left(\frac{\pi y}{A}  \right)   \right)dy\\
&+\frac{\pi}{A}\int_1^{+\infty}\mu_f\left(\frac{\pi}{yA}  \right)(y^{-1}-y^{-2})dy
\end{align*}
and $\displaystyle C_A:=\frac{\pi}{A}\int_1^{+\infty}\mu_f\left(\frac{\pi}{yA}  \right)(y^{-1}-y^{-2})$ is a finite constant which not depends on $L$. Now our sufficient condition is clear because, for any Bravais lattice $L$ of area $A$, we have
$$
E_f[L]-E_f[\Lambda_A]=\frac{\pi}{A}\int_1^{+\infty}\left[\theta_L\left(\frac{y}{2A}  \right)-\theta_{\Lambda_A}\left(\frac{y}{2A}  \right)  \right]g_A(y)dy.
$$
By Montgomery theorem, $\theta_L(u)-\theta_{\Lambda_A}(u)\geq 0$ for any $u\geq 1$ and any $L$. Thus, if
$$
y^{-1}\mu_f\left(\frac{\pi}{yA}\right)+\mu_f\left(\frac{\pi y}{A}  \right)\geq 1 \quad \text{ for a.e. } y\geq 1
$$
then it follows that
$$
E_f[L]-E_f[\Lambda_A]\geq 0
$$
and $\Lambda_A$ is the unique minimizer of $E_f$, up to rotation, among lattices of fixed area $A$. 
\end{proof}

\subsection{Minimization at high density for differentiable inverse Laplace transform}

In this part we give two results, in the case of differentiable inverse Laplace transform, which are based on our Theorem \ref{THM1}.

\begin{prop}\label{PropLapinv}
Let $f$ be an admissible potential such that $\mu_f$ is $C^1$ with derivative $\mu_f'$. If
\begin{enumerate}
\item $\mu_f(y)\geq 0$ on $\displaystyle \left[\frac{\pi}{A},+\infty\right)$,
\item $\displaystyle \mu_f'\left( \frac{\pi}{A}y\right)\geq \frac{1}{y^3}\mu_f'\left(\frac{\pi}{Ay} \right)$ for any $y\geq 1$,
\end{enumerate}
then $\Lambda_A$ is the unique minimizer, of $E_f$, up to rotation,  among Bravais lattices of fixed area $A$.
\end{prop}
\begin{proof}
We write, for any $y\geq 1$, 
$$
g_{A}(y):=y^{-1}\mu_f\left(\frac{\pi}{yA}\right)+\mu_f\left(\frac{\pi y}{A}  \right)=\frac{u_A(y)}{y}
$$
with
$$
u_A(y):=\mu_f\left(\frac{\pi}{yA}\right)+y\mu_f\left(\frac{\pi y}{A}  \right).
$$
Therefore, we get
$$
u_A'(y)=\mu_f\left(\frac{\pi y}{A}  \right)+\frac{\pi y}{A}\left[\mu_f'\left( \frac{\pi}{A}y\right)-y^{-3}\mu_f'\left(\frac{\pi}{Ay} \right]  \right).
$$
Assumption 1. implies that $\mu_f\left(\frac{\pi y}{A}  \right)\geq 0$ for any $y\geq 1$. Moreover, it is clear that point 2. means that $\mu_f'\left( \frac{\pi}{A}y\right)-y^{-3}\mu_f'\left(\frac{\pi}{Ay} \right)\geq 0$ for any $y\geq 1$, hence $u_A'(y)\geq 0$ for any $y\geq 1$. As
$$
u_A(1)=2\mu_f\left(\frac{\pi}{A}  \right)\geq 0
$$
we have $u_A(y)\geq 0$ for any $y\geq 1$ and it follows that
$$
g_A(y)\geq 0 \quad \forall y\geq 1
$$
and by Theorem \ref{THM1}, $\Lambda_A$ is the unique minimizer, up to rotation, of $E_f$ among Bravais lattices of fixed area $A$. 
\end{proof}

\begin{corollary}
If $f$ is an admissible potential such that its inverse Laplace transform $\mu_f$ is convex on $(0,+\infty)$, then there exists $A_0>0$ such that for any $A\in (0,A_0)$, $\Lambda_A$ is the unique minimizer of $E_f$, up to rotation, among Bravais lattices of fixed area $A$.
\end{corollary}
\begin{proof}
As $\mu_f$ is convex, there exists $r_0>0$ such that, for any $r\geq r_0$, $\mu_f(r)\geq 0$. Moreover, for any $y\geq 1$,
$$
\mu_f'\left( \frac{\pi}{A}y\right)\geq \mu_f'\left(\frac{\pi}{Ay} \right)
$$
because $\displaystyle \frac{\pi}{Ay}\leq \frac{\pi y}{A}$ and $\mu_f$ is convex. Hence, as $y^{-3}\leq 1$ for any $y\geq 1$, we get both points 1. and 2. of Proposition \ref{PropLapinv} for any $A$ such that $\displaystyle 0<A\leq A_0:=\frac{\pi}{r_0}$. 
\end{proof}

\subsection{Remarks about our method}

As we saw in \cite{Betermin:2014fy}, for Lennard-Jones case, our method is not optimal to finding all areas such that $\Lambda_A$ is the unique minimizer, up to rotation, of $E_f$ among Bravais lattices of fixed area $A$. The general and difficult main problem is to find all $A$ such that, for any Bravais lattice $L$ of area $A$,
$$
E_f[L]-E_f[\Lambda_A]=\frac{\pi}{A}\int_1^{+\infty}\left[\theta_L\left(\frac{y}{2A}  \right)-\theta_{\Lambda_A}\left(\frac{y}{2A}  \right)  \right]g_A(y)dy\geq 0
$$
where $\displaystyle g_A(y):=y^{-1}\mu_f\left(\frac{\pi}{yA}\right)+\mu_f\left(\frac{\pi y}{A}  \right)$.
We can imagine that even if $g_A$ is not positive almost everywhere on $[1,+\infty)$, the positive part of this integral can compensate the negative one. For instance, if we consider, as in \cite{Betermin:2014fy}, $f(r)=r^{-6}-2r^{-3}$, then
$$
g_A(y)=\frac{\pi^2}{A^2}\left[\frac{\pi^3}{A^35!}(y^6+y^{-5})-y^3-y^{-2}\right]
$$
and we plot graphs of $\displaystyle y\mapsto \frac{\pi^3}{A^35!}(y^6+y^{-5})-y^3-y^{-2} $ for $A=0.8$ (on the left) and $A=1$ (on the right).
\begin{center}
\includegraphics[width=8cm,height=60mm]{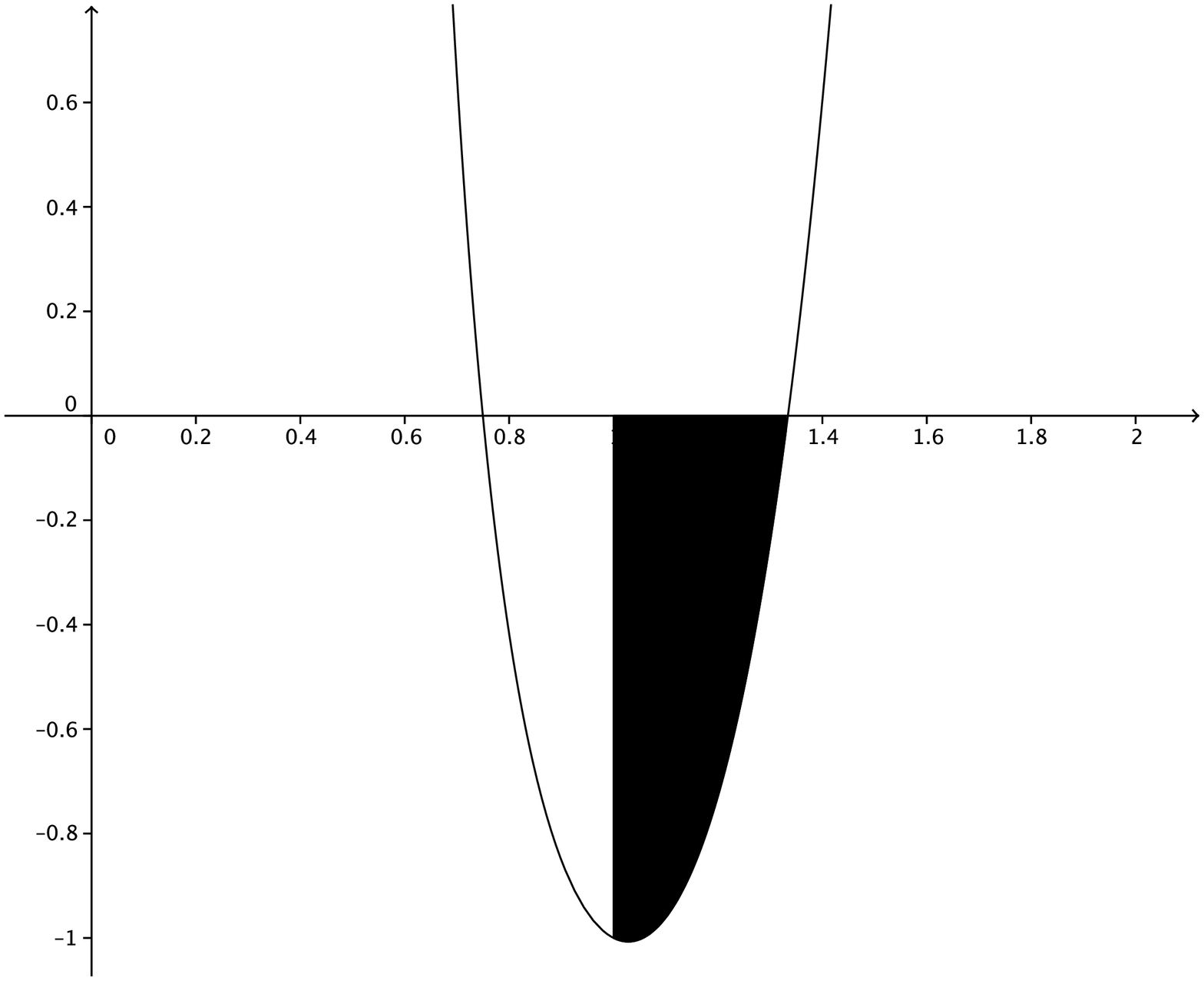}\includegraphics[width=8cm,height=60mm]{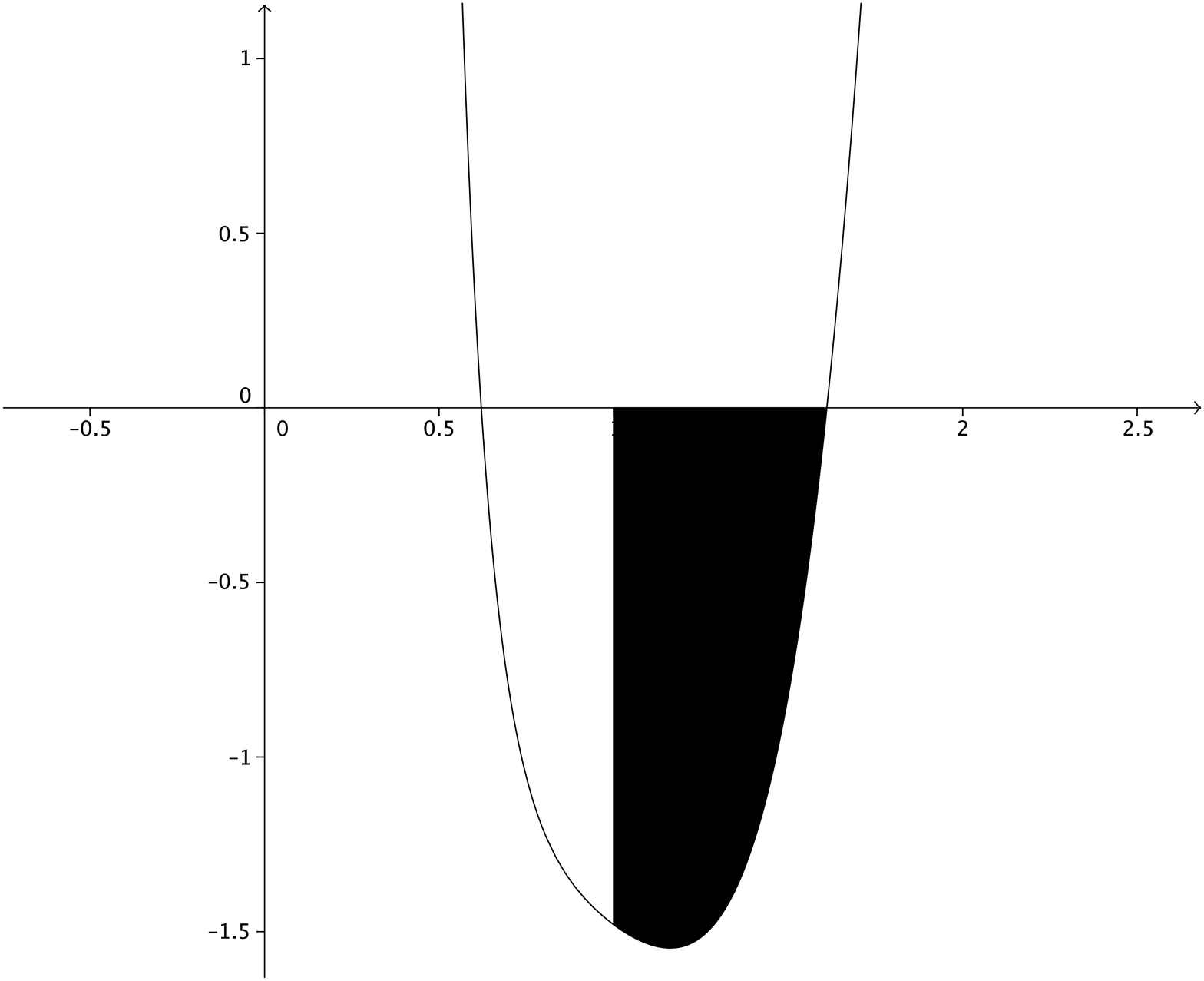}\\
\textbf{Fig. 3 :} Black zone is $\displaystyle \int_1^{y_A}g_A(y)dy$ where $y_A$ is the second zero of $g_A$, $A\in\{ 0.8,1\}$
\end{center}
Thus a fine study, with respect to lattices $L$ and real $y$, of the behaviour of positive function 
$$
\Delta_L(y):=\theta_L\left(\frac{y}{2A}  \right)-\theta_{\Lambda_A}\left(\frac{y}{2A}\right)
$$
is necessary. However we find it difficult at this time. Indeed,
\begin{itemize}
\item for any Bravais lattice $L$ of area $A$, 
$$
\lim_{y\to +\infty} \Delta_L(y)  =0
$$
and $\Delta_L$ exponentially decreases ;
\item if complete monotonicity is a necessary condition to optimality of $\Lambda_A$ for any fixed area $A$, then function $y\mapsto \Delta_L(y) $ is \textbf{not decreasing} on $[1,+\infty)$ for any $A$ and any $L$ with area $A$. Indeed, $\Delta_L$ is decreasing on $(1,+\infty)$ if and only if for any $t\geq 1$, $\Delta_L'(y)\leq 0$, i.e.
\begin{align*}
&\forall A,\forall L,\Delta_L \text{ decreases on } (1,+\infty)\\
 &\iff \forall A,\forall L,\forall y\geq 1,-\frac{\pi}{A}\sum_{p\in L^*}\|p\|^2 e^{-\frac{\pi}{A}y\|p\|^2}+\frac{\pi}{A}\sum_{p\in \Lambda_A^*}\|p\|^2e^{-\frac{\pi}{A}y\|p\|^2}\leq 0\\
&\iff \forall A,\forall L, \forall y\geq 1,\sum_{p\in \Lambda_A^*}\|p\|^2e^{-\frac{\pi}{A}y\|p\|^2}\leq \sum_{p\in L^*}\|p\|^2 e^{-\frac{\pi}{A}y\|p\|^2}
\end{align*}
which would be not possible, because $r\mapsto re^{-\frac{\pi}{A}yr}$ is never completely monotonic for $y\geq 1$. 
\end{itemize}
Hence comparing $\int_1^{y_A}\Delta_{L}(y)g_A(y)dy$ and $\int_{y_A}^{+\infty}\Delta_{L}(y)g_A(y)dy$ seems difficult, even improving our method is possible (see \cite{Betermin:2014fy} for numerical values).

\section{Sums of screened Coulomb potentials}
In this part, we give the first simple example of application of Theorem \ref{THM1}. We consider non convex sums of screened Coulomb potentials and we prove minimality of $\Lambda_A$ at high density and global minimality of a triangular lattice among all Bravais lattices.

\subsection{Definition and proof of Theorem \ref{THM2}.A for $\varphi_{a,x}$}
\begin{defi}
Let $n\in \N^*$. For coefficients $a=(a_1,...,a_n)\in (\R^*)^n$ such that $\displaystyle \sum_{i=1}^n a_i\geq 0$ and for $x=(x_1,...,x_n)\in (\R_+^*)^n$, we define
$$
\varphi_{a,x}(r):=\sum_{i=1}^n a_i\frac{e^{-x_i r}}{r}
$$
and we set $\displaystyle \mathcal{K}_a:=\left\{k ; \sum_{i=1}^ka_i<0\right\}$.
\end{defi}
\noindent As a proof of Theorem \ref{THM2}.A for this potential, we purpose to give an explicit bound for the minimality at high density as follows

\begin{prop}\label{SCGen}
Assume $0<x_1<...<x_n$ and let $A$ such that
$$
A\leq \min\left\{\min_{k\in \mathcal{K}_a}\left\{\frac{\pi}{x_{k+1}}\left(-\frac{\sum_{i=1}^na_i}{\sum_{i=1}^k a_i}  \right)\right\},\frac{\pi}{x_n}\right\}
$$
then $\Lambda_A$ is the unique minimizer of $E_{\varphi_{a,x}}$, up to rotation, among Bravais lattices of fixed area $A$.
\end{prop}
\begin{proof}
We compute easily, because $\mathcal{L}^{-1}[r^{-1}e^{-x_ir}](y)=\mathds{1}_{[x_i,+\infty)}(y)$ for any $x_i>0$ and any $y\geq 0$,
$$
\mu_{\varphi_{a,x}}(y)=\sum_{i=1}^n a_i \mathds{1}_{[x_i,+\infty)}(y)
$$
and it follows that, for any $y\geq 1$,
\begin{align*}
g_A(y)&:=\frac{1}{y}\sum_{i=1}^n\mathds{1}_{[x_i,+\infty)}\left(\frac{\pi}{yA}  \right)+\sum_{i=1}^n a_i\mathds{1}_{[x_i,+\infty)}\left( \frac{\pi y}{A} \right)\\
&=\frac{1}{y}\sum_{i=1}^n a_i \mathds{1}_{\left[1,\frac{\pi}{Ax_i}\right]}(y)+\sum_{i=1}^n a_i\mathds{1}_{\left[  \frac{Ax_i}{\pi} ,+\infty\right)}(y).
\end{align*}
As, by assumption, $\displaystyle A\leq \frac{\pi}{x_n}$, we have, for any $1\leq i\leq n-1$, 
$$
\frac{\pi}{Ax_i}\geq \frac{\pi}{Ax_{i+1}}\geq 1.
$$ 
Hence we get 
\[ g_A(y)=
\left\{
\begin{array}{ll}
\displaystyle (1+y^{-1})\sum_{i=1}^n a_i &\mbox{if $1\leq y\leq \frac{\pi}{Ax_n}$}\\
\displaystyle \sum_{i=1}^k\frac{a_i}{y}+\sum_{i=1}^na_i &\mbox{if $\frac{\pi}{Ax_{k+1}}<y\leq \frac{\pi}{Ax_k}$, for any $1\leq k\leq n-1$}\\
\displaystyle \sum_{i=1}^n a_i &\mbox{if $y>\frac{\pi}{A x_1}.$}
\end{array}\right.\]
As $\sum_{i=1}^n a_i\geq 0$ and, for any $k\not\in \mathcal{K}_a$, $\sum_{i=1}^k a_i \geq 0$, we obtain 
$$
\forall y\in \left[1,\frac{\pi}{Ax_n}\right]\bigcup_{k\not\in \mathcal{K}_a}\left(\frac{\pi}{Ax_{k+1}},\frac{\pi}{Ax_k} \right]\cup \left(\frac{\pi}{Ax_1},+\infty  \right), \quad g_A(y)\geq 0.
$$
Now if $k\in \mathcal{K}_a$, as $\displaystyle A\leq \min_{k\in \mathcal{K}_a}\left\{\frac{\pi}{x_{k+1}}\left(-\frac{\sum_{i=1}^na_i}{\sum_{i=1}^k a_i}  \right)\right\}$, we get, for any $y\in \left(\frac{\pi}{Ax_{k+1}},\frac{\pi}{Ax_k}  \right]$,
$$
\sum_{i=1}^k\frac{a_i}{y}+\sum_{i=1}^na_i \geq \frac{Ax_{k+1}}{\pi}\sum_{i=1}^k a_i+\sum_{i=1}^na_i\geq 0
$$
and it follows that $g_A(y)\geq 0$ for any $y\geq 1$. By Theorem \ref{THM1}, $\Lambda_A$ is the unique minimizer of $E_{\varphi_{a,x}}$, up to rotation, among Bravais lattices of fixed area $A$.
\end{proof}

\subsection{Global minimality : Proof of Theorem \ref{THM2}.B.1}
Now we focus on particular ``attractive-repulsive" case 
\begin{itemize}
\item $a=(-a_1,a_2)$ where $0<a_1<a_2$;
\item $x=(x_1,x_2)$ with $0<x_1<x_2$.
\end{itemize}
Therefore we define, for any $y>0$,
$$
\varphi_{a,x}^{AR}(r):=a_2\frac{e^{-x_2r}}{r}-a_1\frac{e^{-x_1r}}{r}.
$$
Now, let us prove Theorem \ref{THM2}.B.1.
\begin{proof}
Firstly we study variations of $\varphi_{x,a}$ to prove the existence of global minimizer $L_{a,x}$ among all Bravais lattices and upper bound $\alpha_{a,x}$ for its area. Afterward we prove that inequalities \eqref{CondMinGlob} are equivalent with
\begin{equation*}\label{alpha}
\alpha_{a,x}\leq \min\left\{\frac{\pi}{x_2},\frac{\pi}{x_2}\left( \frac{a_2}{a_1}-1 \right)  \right\}.
\end{equation*}
Thus, by direct application of Theorem \ref{SCGen}, if $A\leq \min\left\{\frac{\pi}{x_2},\frac{\pi}{x_2}\left( \frac{a_2}{a_1}-1 \right)  \right\}$, $\Lambda_A$ is the unique minimizer among Bravais lattices of fixed area $A$, therefore $L_{a,x}$ is triangular and unique. \\ \\

\textbf{STEP 1 : Variations of function $\varphi_{a,x}$}\\
We have, for any $r>0$, 
$$
\varphi_{a,x}'(r)= \frac{1}{r^2}\left[ a_1(1+x_1r)e^{-x_1r}-a_2(1+x_2r)e^{-x_2r} \right]
$$
and it follows that
$$
\varphi_{a,x}'(r)\geq 0 \iff g_{a,x}(r):=(x_2-x_1)r+\ln(1+x_1r)-\ln(1+x_2r)+\ln\left( \frac{a_1}{a_2} \right)\geq 0.
$$
As, for any $r>0$,
$$
g_{a,x}'(r)=\frac{(x_2-x_1)\left(x_1x_2r^2+(x_1+x_2)r\right)}{(1+x_1r)(1+x_2r)}> 0,
$$
$g_{a,x}$ is an increasing function on $(0,+\infty)$. We have $a_2>a_1$, therefore $\ln\left( \frac{a_1}{a_2}\right)<0$ and there exists $\alpha_{a,x}$ such that 
$$
\forall r\in (0,\alpha_{a,x}], g_{a,x}(r)\leq 0, \quad \text{ and } \quad \forall r>\alpha_{a,x}, g_{a,x}(r)>0.
$$
Thus we get $\varphi_{a,x}$ is a decreasing function on $(0,\alpha_{a,x}]$ and an increasing function on $(\alpha_{a,x},+\infty)$.\\ \\

\textbf{STEP 2 : The existence of global minimizer for $E_{\varphi_{a,x}}$}\\
Variations of function $\varphi_{a,x}$ and the fact that $\displaystyle \lim_{r\to 0 \atop r>0}\varphi_{a,x}(r)=+\infty$ and goes to $0$ at infinity implies that global minimizer exists. Indeed, this problem can be viewed like a minimization problem of a three variables function. By previous limits we can restrict this problem with variables in a compact set, and by continuity this problem has a solution $L_{a,x}$.\\ \\

\textbf{STEP 3 : Upper bound for $|L_{a,x}|$ and conclusion}\\
Let $L_{a,x}=\Z u_{a,x}\oplus \Z v_{a,x}$. 
If $\|u_{a,x}\|>\sqrt{\alpha_{a,x}}$ then a contraction of all distances yields a new lattice with smaller energy because, by STEP 1, $r\mapsto \varphi_{a,x}(r^2)$ is an increasing function on $\displaystyle (\sqrt{\alpha_{a,x}},+\infty)$. Moreover, if $\|v_{a,x}\|>\sqrt{\alpha_{a,x}}$ then a contraction of $\R v_{a,x}$ also gives a lattice with less energy. Thus we have $\|u_{a,x}\|\leq \|v_{a,x}\|\leq \sqrt{\alpha_{a,x}}$. Now, because $|L_{a,x}|\leq \|u_{a,x}\|\|v_{a,x}\|$, we get\footnote{This argument appears in \cite[Proposition 4.1, ii)]{Betermin:2014fy} and in \cite{MEKBS} in order to prove that the distance between two animals in a swarm is less than a specific ``confort distance" between them, which minimizes a certain function.}
$$
|L_{a,x}|\leq \alpha_{a,x}.
$$
Now it is not difficult to check that
$$
\varphi_{a,x}'\left( \frac{\pi}{x_2} \right)\geq 0 \iff \frac{a_1\left(1+\frac{x_1}{x_2}\pi  \right)}{a_2(1+\pi)}e^{\left( 1-\frac{x_1}{x_2} \right)\pi}\geq 1
$$
and
$$
\varphi_{a,x}'\left( \frac{\pi}{x_2}\left( \frac{a_2}{a_1}-1 \right)  \right)\geq 0 \iff  \frac{a_1\left( a_1x_2+x_1(a_2-a_1)\pi \right)}{a_2x_2\left(a_1+(a_2-a_1)\pi\right)}e^{\left(1-\frac{x_1}{x_2}  \right)\left( \frac{a_2}{a_1}-1 \right)\pi}\geq 1
$$
hence \eqref{alpha} holds and $L_{a,x}$ is unique and triangular by Theorem \ref{SCGen} as explained at the beginning of the proof.\\ \\

\textbf{STEP 4 : Example}\\
If we take $a_2=2a_1$ then
$$
\frac{a_1\left( a_1x_2+x_1(a_2-a_1)\pi \right)}{a_2x_2\left(a_1+(a_2-a_1)\pi\right)}e^{\left(1-\frac{x_1}{x_2}  \right)\left( \frac{a_2}{a_1}-1 \right)\pi}=\frac{a_1\left(1+\frac{x_1}{x_2}\pi  \right)}{a_2(1+\pi)}e^{\left( 1-\frac{x_1}{x_2} \right)\pi}=\frac{1}{2(1+\pi)}(1+\frac{x_1}{x_2}\pi)e^{\left( 1-\frac{x_1}{x_2} \right)\pi}.
$$
Now we set $X=\frac{x_1}{x_2}\pi$ and our condition becomes $\displaystyle \frac{(1+X)}{2(1+\pi)}e^{-X+\pi}\geq 1$, which is equivalent with
$$
g(X):=-X+\log(1+X)-\log(2+2\pi)+\pi\geq 0.
$$
As $g'(X)=-\frac{X}{1+X}\geq 0$ on $\R_+$, then $g$ decreases and there exists $\tilde{X}>0$ such that $g(\tilde{X})=0$. Numerically, we found $\tilde{X}>2.186$, hence if $X\leq 2.186$, which corresponds to $\displaystyle \frac{x_1}{x_2}\pi\leq 2.186$, i.e. $\displaystyle x_1\leq \frac{2.186}{\pi}x_2\approx 0.695825x_2$,
then $g(X)\geq 0$. In particular, it is true if $x_1\leq 0.695 x_2$.
\end{proof}

\begin{example}
For instance, we can choose $(x_1,x_2)=(1,2)$. Thus, global minimizer of
$$
L\mapsto E_{\varphi_{a,x}}[L]=\sum_{p\in L^*}\varphi_{a,x}(\|p\|^2)=2a_1\sum_{p\in L^*}\frac{e^{-2\|p\|^2}}{\|p\|^2}-a_1\sum_{p\in L^*}\frac{e^{-\|p\|^2}}{\|p\|^2}
$$
is unique, up to rotation, and triangular. Hence we can construct potential with arbitrary deep well (using parameter $a_1$) and with triangular global minimizer.
\begin{center}
\includegraphics[width=10cm,height=80mm]{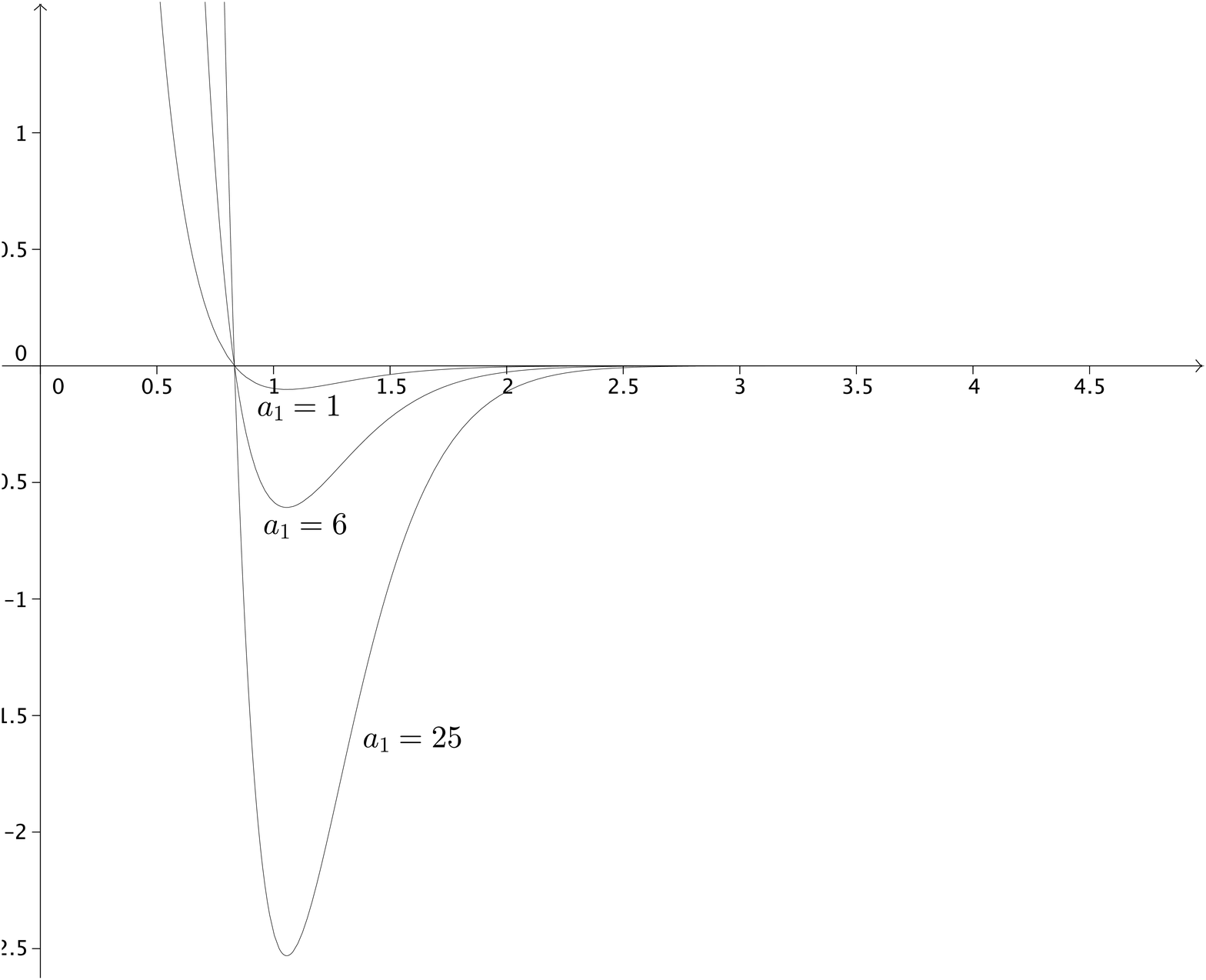}\\
\textbf{Fig. 4 :} Graph of $\displaystyle r\mapsto \varphi_{a,x}(r^2)= 2a_1\frac{e^{-2r^2}}{r^2}-a_1\frac{e^{-r^2}}{r^2}$ for $a_1\in \{1,6,25\}$.
\end{center}
\end{example}
\begin{remark}
This kind of potential seems not to be used in molecular simulation but this prediction of triangular ground state could be observed in the future. Furthermore our Theorem \ref{THM2}.B.1 allows to better understand ground state for parametrized potential with repulsion at short distance and quick decay at large distance, as in \cite{Crystal} where Theil proved global minimality of a triangular lattice among all configurations if the potential's well is sufficiently narrow, i.e. with repulsion and decay sufficiently strong.
\end{remark}

\section{Nonconvex sums of inverse power laws}

In this part, we generalize our result in \cite{Betermin:2014fy}, which tackled only the classical Lennard-Jones case, for any nonconvex sums of inverse power potentials, that is to say optimality of triangular lattice $\Lambda_A$ at for high densities and non-optimality of this one for low densities. Furthermore we show that our method allows to obtain global minimizer, i.e. minimizer among all Bravais lattices without constraint of area, of Lennard-Jones type energies with small parameters.\\

\subsection{Definition and proof of Theorem \ref{THM2}.A for $V_{a,x}$}

\begin{defi}
Let $n\geq 1$ be an integer and, for $a=(a_1,...,a_n)\in (\R^*)^n$ such that $a_n>0$, and $x=(x_1,...,x_n) \in (\R_+)^n$ such that $1<x_1<...<x_n$, let
$$
V_{a,x}(r)=\sum_{i=1}^n \frac{a_i}{r^{x_i}}.
$$
We set $I_-:=\{i; a_i<0 \}$, $I_+:=\{i; a_i>0 \}$ and $\displaystyle \alpha_i:=\frac{a_i \pi^{x_i-1}}{\Gamma(x_i)}$. Moreover we assume that $I_-\neq \emptyset$ (otherwise $V_{a,x}$ is completely monotonic).
\end{defi}

\begin{remark} In order to minimize $E_{V_{a,x}}$ among lattices, we should assume $a_n>0$ because we have $V_{a,x}(r)\sim a_n r^{-x_n}$ as $r\to 0$. Indeed we have $V_{a,x}(r)\to +\infty$ as $r\to 0$ and $V_{a,x}(r)\to 0$ as $r\to +\infty$ therefore there exists minimizer of $E_{V_{a,x}}$ among Bravais lattices with fixed area. If $a_n<0$, it is sufficient to do $\|u\|\to 0$ to get $E_{V_{a,x}}[L]\to-\infty$.
\end{remark}

\begin{example}
This kind of potential is widely used in molecular simulation. Indeed, besides Lennard-Jones potentials that we will study in the next subsection, it is sometimes necessary to consider some modifications of it. For instance, the $(12-6-4)$ potential proposed by Mason and Schamp in \cite{MasonSchamp}, defined by
$$
V(r)=\frac{a_3}{r^{12}}-\frac{a_2}{r^{6}}-\frac{a_1}{r^4},
$$
describes the interaction of ions with neutral systems. For instance, in fullerene $C_{60}$, this potential describes interaction between a carbon atom in the polyatomic ion and a buffer gas helium atom.\\
An other example, proposed by Klein and Hanley in \cite{KHanley70,KHanley72} for description of rare gases, more precise than Lennard-Jones, is the potential defined, for $m>8$, by
$$
V(r)=\frac{a_3}{r^m}-\frac{a_2}{r^6}-\frac{a_1}{r^8}.
$$
\end{example}

\noindent As in previous section, we give an explicit bounds for the minimality of $\Lambda_A$ at high density in the following proposition.

\begin{prop}\label{HDInvPower}
If we have
\begin{equation}\label{HD}
A\leq \pi \min_{i\in I_-}\left\{  \left( \frac{a_n\Gamma(x_i)}{2\sharp\{I_-\}|a_i|\Gamma(x_n)} \right)^{\frac{1}{x_n-x_i}}\right\}
\end{equation}
then $\Lambda_A$ is the unique minimizer of $E_{V_{a,x}}$, up to rotation, among Bravais lattices of fixed area $A$.
\end{prop}
\begin{proof}
By usual formula, we have
$$
\mu_{V_{a,x}}(y)=\sum_{i=1}^n \frac{a_i}{\Gamma(x_i)}y^{x_i-1}
$$
and it follows that
\begin{align*}
g_A(y)&:=y^{-1}\mu_{V_{a,x}}\left(\frac{\pi}{yA}\right)+\mu_{V_{a,x}}\left(\frac{\pi y}{A}  \right)=\sum_{i=1}^n \frac{\alpha_i}{A^{x_i-1}}(y^{-x_i}+y^{x_i-1})\\
&=y^{-x_n}\sum_{i=1}^n\frac{\alpha_i}{A^{x_i-1}}(y^{x_n-x_i}+y^{x_n+x_i-1}).
\end{align*}
We set 
$$
p_{a,x}(y):=\sum_{i=1}^n\frac{\alpha_i}{A^{x_i-1}}(y^{x_n-x_i}+y^{x_n+x_i-1}).
$$
We notice that the term of high order is $\frac{\alpha_n}{A^{x_n-1}}y^{2x_n-1}$ with $\alpha_n>0$ and the number of negative coefficients is $2\sharp\{I_-\}$. Thus, by Cauchy's rule \ref{Cauchy} and more precisely generalization \ref{UBoundRoot}, an upper bound on the values of the positive zero of $p_{a,x}$ is
$$
M_{p_{a,x}}:=\max_{i\in I_-}\left\{\left( \frac{2\sharp\{I_-\}|\alpha_i|A^{x_n-x_i}}{\alpha_n} \right)^{\frac{1}{x_n-x_i}},\left(\frac{2\sharp\{I_-\}|\alpha_i|A^{x_n-x_i}}{\alpha_n}  \right)^{\frac{1}{x_n+x_i-1}}\right\}.
$$
because $2x_n-1-(x_n-x_i)=x_n+x_i-1$ and $2x_n-1-(x_n+x_i-1)=x_n-x_i$.\\ 
We notice that 
\begin{align*}
& A\leq \pi \min_{i\in I_-}\left\{\left(\frac{a_n \Gamma(x_i)}{2\sharp\{I_-\}|a_i| \Gamma(x_n)}  \right)^{\frac{1}{x_n-x_i}}  \right\} = \min_{i\in I_-}\left\{\left(\frac{\alpha_n}{2\sharp\{I_-\}|\alpha_i|}  \right)^{\frac{1}{x_n-x_i}}  \right\}\\
&\iff A\leq \left(\frac{\alpha_n}{2\sharp\{I_-\}|\alpha_i|}  \right)^{\frac{1}{x_n-x_i}}, \quad \forall i\in I_- \\
&\iff \frac{2A^{x_n-x_i}\sharp\{I_-\}|\alpha_i|}{\alpha_n}\leq 1, \quad \forall i\in I_-\\
&\iff M_{p_{a,x}}\leq 1
\end{align*}
therefore the assumption implies that the largest zero of $p_{a,x}$ is less than $1$. As $\alpha_n>0$, it follows that $p_{a,x}(y)\geq 0$ for any $y\geq M_{p_{a,x}}$ and then $g_A(y)\geq 0$ for any $y\geq 1$ and by Theorem \ref{THM1}, if \eqref{HD} holds, then $\Lambda_A$ is the unique minimizer of $E_{V_{a,x}}$ among Bravais lattices of fixed area $A$.
\end{proof}

\begin{remark}
This result seems to be natural because for $r$ close to $0$, $V_{a,x}(r)\sim a_n r^{-x_n}$ and for any $A$, $\Lambda_A$ is the unique minimizer of $L\mapsto \zeta_L(2x_n)$ among Bravais lattices of fixed area $A$. However, if we fix $A$, $\|u\|$ and $\|v\|$ can be as larger as we want and the behavior of $V_{a,x}$ can be unusual.\\
Furthermore, in the case 
$$
V_{a,x}(r)=\frac{a_1}{r^{x_1}}+\frac{a_2}{r^{x_2}}+\frac{a_3}{r^{x_3}}
$$ 
where $a_1,a_3$ are positive and $a_2$ negative, our bound \eqref{HD} does not depend on $a_1$. For instance, if $a=(p,-3,1)$ and $x=(2,4,6)$, then, for any $p$,
$$
\pi \min_{i\in I_-}\left\{\left(\frac{a_n \Gamma(x_i)}{\sharp\{I_-\}|a_i| \Gamma(x_n)}  \right)^{\frac{1}{x_n-x_i}}  \right\}=\pi\left(\frac{\Gamma(4)}{6\Gamma(6)}  \right)^{1/2}\approx 0.2867869
$$
which corresponds to triangular lattices of length $\approx 0.5754589$.
\begin{center}
\includegraphics[width=8cm,height=60mm]{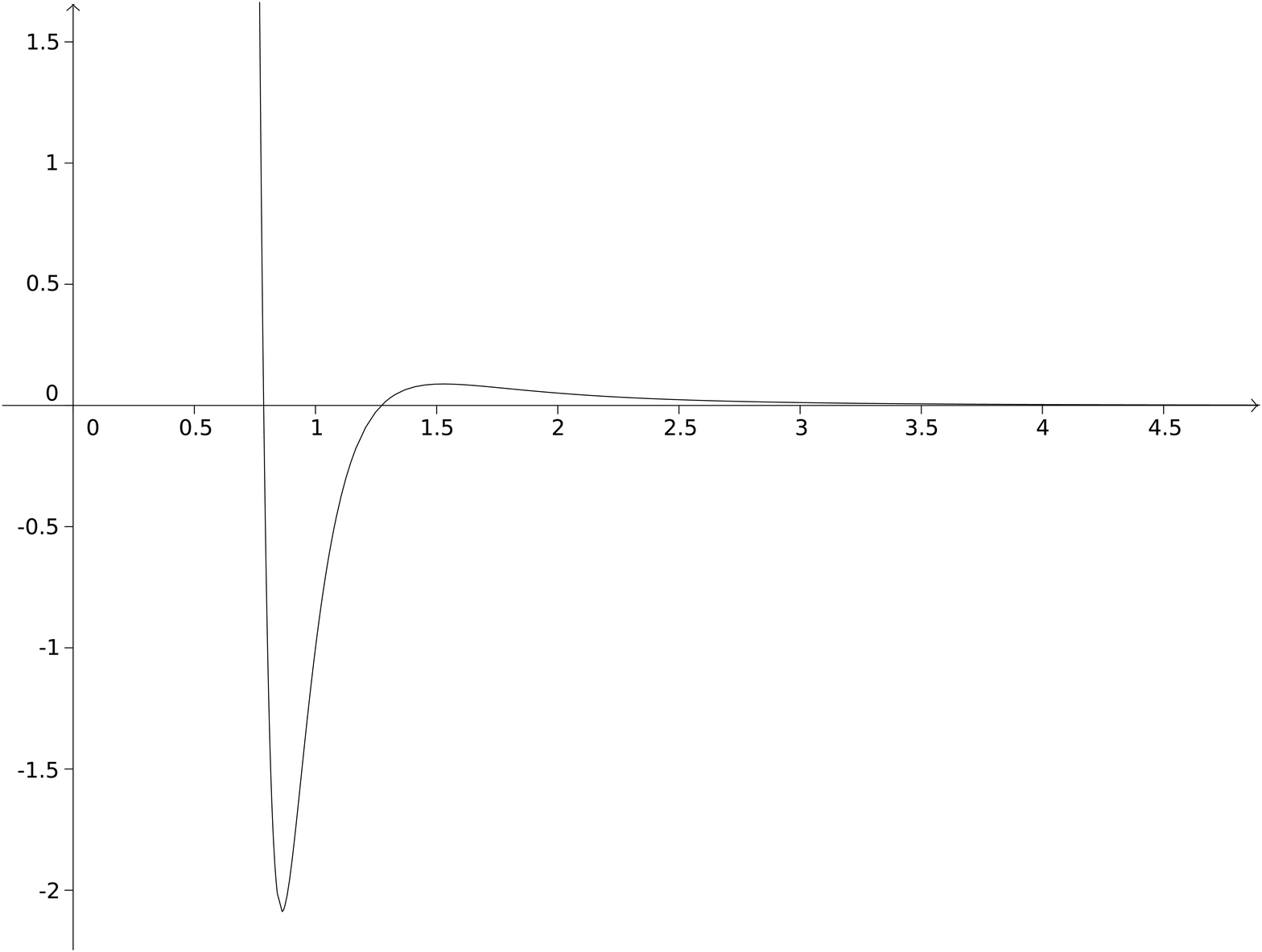}\includegraphics[width=8cm,height=60mm]{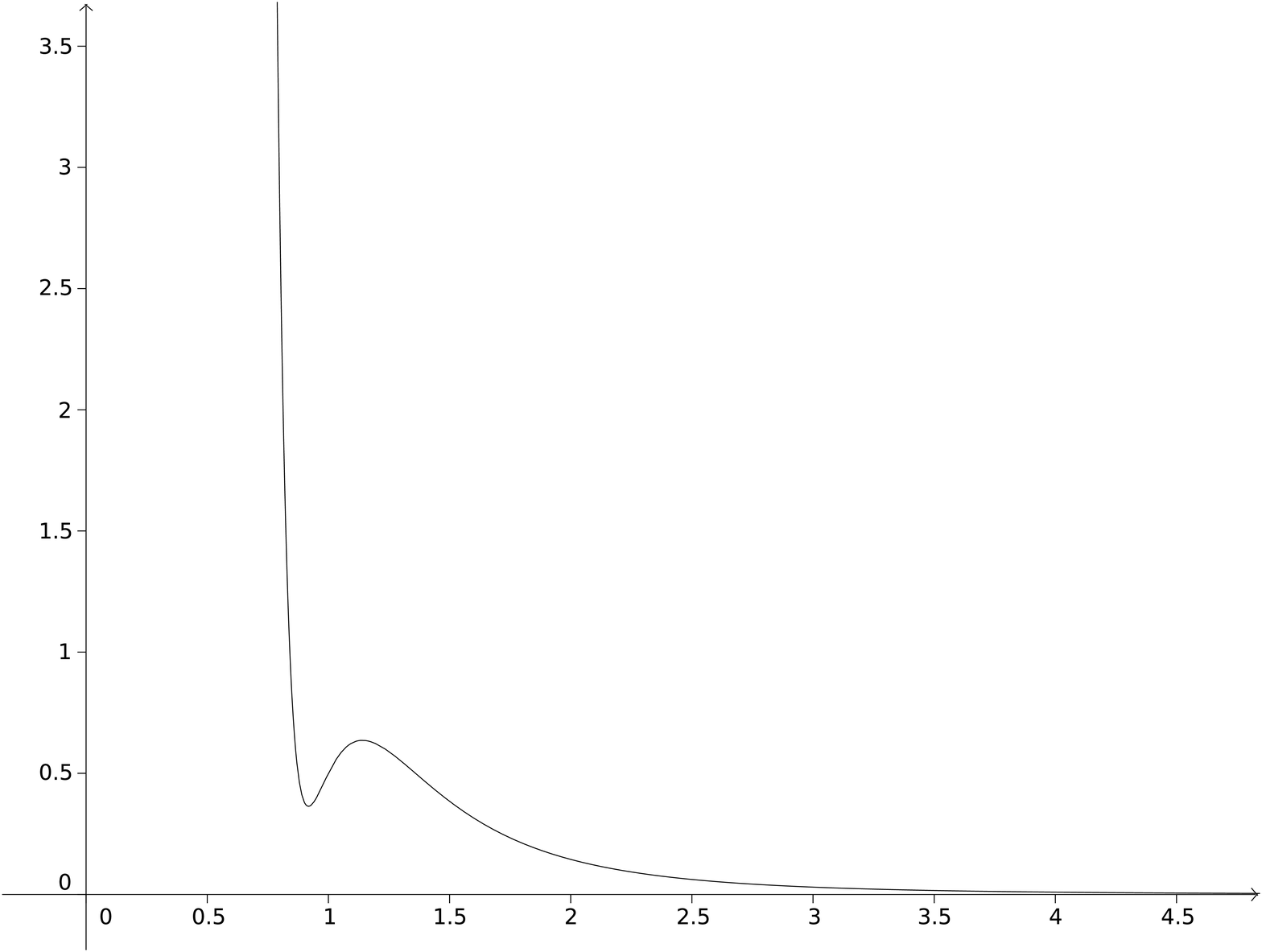}\\
\textbf{Fig. 5 :} Graphs of $\displaystyle V_{a,x}(r^2)=\frac{p}{r^4}-\frac{3}{r^8}+\frac{1}{r^{12}}$ for $p=1$ (on the left) and $p=2.5$ (on the right)
\end{center}
\end{remark}

\begin{example}
For our counterexample \eqref{CE}, i.e. $\displaystyle V(r)=\frac{14}{r^2}-\frac{40}{r^3}+\frac{35}{r^4}$, $a=(14,-40,35)$, $x=(2,3,4)$ and $\sharp\{I_-\}=1$, hence
$$
\pi \min_{i\in I_-}\left\{\left(\frac{a_n \Gamma(x_i)}{\sharp\{I_-\}|a_i| \Gamma(x_n)}  \right)^{\frac{1}{x_n-x_i}}  \right\}=\pi\left(\frac{35\Gamma(3)}{80\Gamma(4)}  \right)^1=\frac{7\pi}{48}\approx 0.4581488,
$$
which corresponds to triangular lattice of length $\approx 0.7273408$. Thus, for $\displaystyle A\leq \frac{7\pi}{48}$, $\Lambda_A$ is the unique minimizer of $E_V$, up to rotation, among Bravais lattices of fixed area $A$.
\end{example}

\subsection{Non-optimality of $\Lambda_A$ at low density}
\begin{prop}  \label{LDInvPower} If $a_1<0$ and 
\begin{equation} \label{LD}
A\geq \inf_{L\neq \Lambda_1 \atop |L|=1} \max_{i\in I_+} \left\{\left( \frac{\sharp\{I_+\}a_i (\zeta_{L}(2x_i)-\zeta_{\Lambda_1}(2x_i))}{|a_1|(\zeta_{L}(2x_1)-\zeta_{\Lambda_1}(2x_1))} \right)^{\frac{1}{x_n-x_i}}  \right\},
\end{equation}
that is to say if $A$ is sufficiently large, then $\Lambda_A$ is not a minimizer of $E_{V_{a,x}}$ among Bravais lattices of fixed area $A$.
\end{prop}
\begin{proof}
Let $L_A=\sqrt{A} L_1$ be a Bravais lattice of area $A$, with $|L_1|=1$, then
\begin{align*}
E_{V_{a,x}}[\Lambda_A]-E_{V_{a,x}}[L_A]&=\sum_{i=1}^n a_i (\zeta_{\Lambda_A}(2x_i)-\zeta_{L_A}(2x_i))\\
&= \sum_{i=1}^n \frac{a_i}{A^{x_i}} (\zeta_{\Lambda_1}(2x_i)-\zeta_{L_1}(2x_i))\\
&=A^{-x_n}\sum_{i=1}^n a_i (\zeta_{\Lambda_1}(2x_i)-\zeta_{L_1}(2x_i))A^{x_n-x_i}.
\end{align*}
We set 
$$
p_{a,x,L_1}(A):=\sum_{i=1}^n a_i (\zeta_{\Lambda_1}(2x_i)-\zeta_{L_1}(2x_i))A^{x_n-x_i}.
$$
As $a_1<0$ and, for any $s>1$, 
$$
\zeta_{\Lambda_1}(2s)-\zeta_{L_1}(2s)\leq 0,
$$ 
because $\Lambda_1$ is the unique minimizer of $L\mapsto \zeta_L(2s)$ among Bravais lattices of area $1$, we can apply Cauchy's rule \ref{Cauchy}, and more precisely its generalization \eqref{UBoundRoot}. The number of negative coefficient of $p_{a,x,L_1}$ is exactly $\sharp\{I_+\}$ and an upper bound on the values of the positive zero of $p_{a,x,L_1}$ for given $L_1$, is
$$
M_{p_{a,x}}(L_1):=\max_{i\in I_+}\left\{\left( \frac{\sharp\{I_+\}a_i (\zeta_{L_1}(2x_i)-\zeta_{\Lambda_1}(2x_i))}{|a_1|(\zeta_{L_1}(2x_1)-\zeta_{\Lambda_1}(2x_1))} \right)^{\frac{1}{x_n-x_i}}  \right\}.
$$
Hence, for any $L$ such that $|L|=1$, if $A\geq M_{p_{a,x}}(L)$ then $p_{a,x,L}(A)\geq 0$. We conclude that if \eqref{LD} holds, then $E_{V_{a,x}}[\Lambda_A]-E_{V_{a,x}}[L_A]\geq 0$ and $\Lambda_A$ cannot be a minimizer of $E_{V_{a,x}}$ among Bravais lattices of fixed area $A$.
\end{proof}

\begin{remark}
To compute explicitly a lower bound for $A$ such that $\Lambda_A$ is not a minimizer of energy $E_{V_{a,x}}$, we can take $L=\Z^2$ in \eqref{LD} and use equalities \eqref{zetasquare} and \eqref{zetatriang} (see next subsection for computations in Lennard-Jones case).
\end{remark}

\subsection{Lennard-Jones type potentials : proofs of Theorems \ref{THM2}.A and \ref{THM2}.B.2 and numerical results}
Now we want to study more precisely the class of Lennard-Jones type potential. In \cite{Betermin:2014fy} we studied classical $(12-6)$ Lennard-Jones potential $V_{LJ}(r)=r^{-12}-2r^{-6}$, such that its minimizer is $1$, and we proved that the minimizer of its energy among lattices with fixed area $A$ is triangular for small $A$ and it cannot be triangular for large $A$. Here we prove that our method gives interesting results for this kind of potential.\\ \\
Let $1<x_1<x_2$ and $a_1,a_2\in (0,+\infty)$, we define \textbf{Lennard-Jones type potential} by 
$$
V_{a,x}^{LJ}(r):=\frac{a_2}{r^{x_2}}-\frac{a_1}{r^{x_1}}, \quad \forall r>0.
$$
\begin{example}
We can cite various Lennard-Jones type potentials used in molecular simulation or in the study of social aggregation (see \cite{MEKBS}), besides the classical $V_{LJ}$. For instance the $(12-10)$ potential
$$
V(r)=\frac{a_2}{r^{12}}-\frac{a_1}{r^{10}}
$$
describes hydrogen bonds (see \cite{GelinKarplus}).\\
A $(6-4)$ potential 
$$
V(r)=\frac{a_2}{r^{6}}-\frac{a_1}{r^{4}}
$$
is also used for finding energetically favourable regions in protein binding sites (see \cite{Goodford} for details).
\end{example}

\begin{lemma}\label{VarLJ}
Let $1<x_1<x_2$, then function $r\mapsto V_{a,x}^{LJ}(r^2)$ is decreasing on $\displaystyle \left[0, \left( \frac{a_2x_2}{a_1x_1} \right)^{\frac{1}{2(x_2-x_1)}}\right)$ and increasing on $\displaystyle \left[\left(\frac{a_2x_2}{a_1x_1} \right)^{\frac{1}{2(x_2-x_1)}},+\infty\right)$.
\end{lemma}
\begin{proof}
The first derivative of this function is $r\mapsto -2a_2x_2 r^{-2x_2-1}+2a_1x_1r^{-2x_1-1}$ and 
$$
\displaystyle (V_{a,x}^{LJ})'(r)\geq 0 \iff r\geq \left( \frac{a_2x_2}{a_1x_1} \right)^{\frac{1}{2(x_2-x_1)}}.
$$
\end{proof}
\begin{center}
\includegraphics[width=10cm,height=80mm]{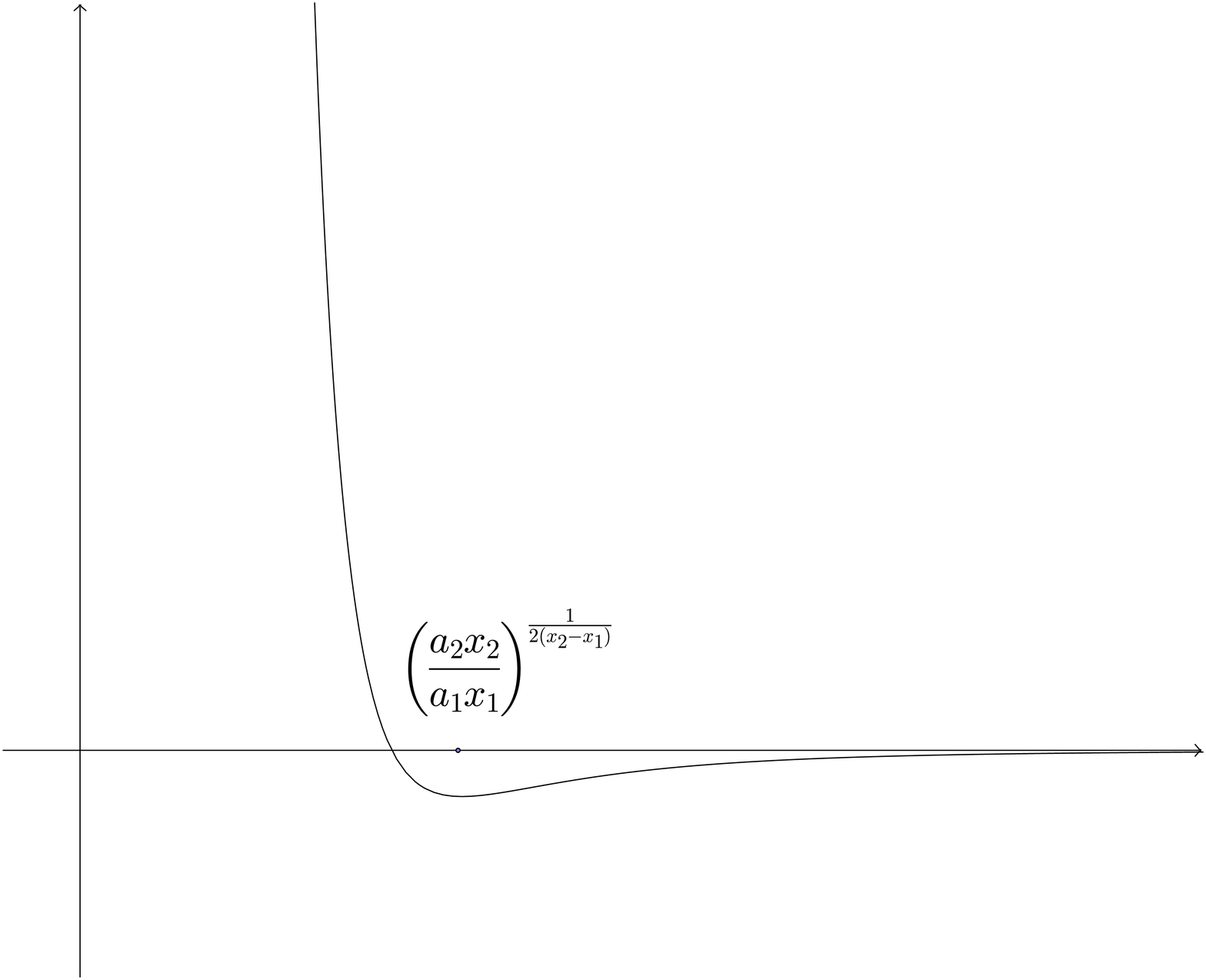}\\
\textbf{Fig. 6 :} Graph of $r\mapsto V_{a,x}^{LJ}(r^2)$
\end{center}
Obviously, the form of potential $V_{a,x}^{LJ}$ implies that minimizer among lattices exists. Indeed, if we fix the area, one of the distance in the lattice cannot be too small otherwise lattice energy goes to infinity (see \cite[Proposition 2.3]{Betermin:2014fy} for details).\\

\noindent As in our previous work \cite{Betermin:2014fy}, the following upper bound for area such that triangular lattice is the unique minimizer for our energy is not optimal but the best for our method. Moreover, its upper bound is better than we apply Cauchy's rule (Proposition \ref{HDInvPower}) but the method is specific for this kind of potential.\\

\begin{prop}\label{LJones} \textbf{(Lennard-Jones at high density)}
If $\displaystyle A\leq \pi\left(\frac{a_2\Gamma(x_1)}{a_1\Gamma(x_2)} \right)^{\frac{1}{x_2-x_1}}$, then $\Lambda_A$ is the unique minimizer of $E_{V_{a,x}^{LJ}}$, up to rotation, among lattices of area $A$ fixed.
\end{prop}
\begin{proof}
We have, by proof of Theorem \ref{HDInvPower}, for any $y\geq 1$, 
\begin{align*}
g_A(y)&=\frac{\alpha_2}{A^{x_2-1}}\left(y^{-x_2}+y^{x_2-1}  \right)-\frac{\alpha_1}{A^{x_1-1}}\left(y^{-x_1}+y^{x_1-1}   \right)=\frac{y^{-x_2}}{A^{x_1-1}}\tilde{g}_A(y)
\end{align*}
where $\displaystyle \tilde{g}_A(y)=\frac{\alpha_2}{A^{x_2-x_1}}y^{2x_2-1}-\alpha_1 y^{x_2+x_1-1}-\alpha_1 y^{x_2-x_1}+\frac{\alpha_2}{A^{x_2-x_1}}$. We compute
\begin{align*}
\tilde{g}'_A(y)=\frac{(2x_2-1)\alpha_2}{A^{x_2-x_1}}y^{2x_2-2}-\alpha_1(x_2+x_1-1)y^{x_2+x_1-2}-\alpha_1(x_2-x_1)y^{x_2-x_1-1}=y^{x_2-x_1-1}u_A(y)
\end{align*}
where $\displaystyle u_A(y)=\frac{(2x_2-1)\alpha_2}{A^{x_2-x_1}}y^{x_2+x_1-1}-\alpha_1(x_2+x_1-1)y^{2x_1-1}-\alpha_1(x_2-x_1)$. Moreover
$$
u'_A(r)=(x_2+x_1-1)y^{2x_1-2}\left[ \frac{(2x_2-1)\alpha_2}{A^{x_2-x_1}}y^{x_2-x_1}-\alpha_1(2x_1-1) \right].
$$
We have $\displaystyle u'_A(\bar{y})=0 \iff \bar{y}=\left( \frac{\alpha_1(2x_1-1)A^{x_2-x_1}}{\alpha_2(2x_2-1)} \right)^{\frac{1}{x_2-x_1}}=\frac{A}{\pi}\left( \frac{a_1\Gamma(x_2)}{a_2\Gamma(x_1)} \right)^{\frac{1}{x_2-x_1}}\left( \frac{2x_1-1}{2x_2-1} \right)^{\frac{1}{x_2-x_1}}$.\\
If $\displaystyle A\leq \pi\left(\frac{a_2\Gamma(x_1)}{a_1\Gamma(x_2)} \right)^{\frac{1}{x_2-x_1}}$ then $\bar{y}<1$ and $u'_A(y)>0$ on $[1;+\infty)$, i.e. $u_A$ is an increasing function on $[1;+\infty)$. Furthermore we have
\begin{align*}
u_A(1)=(2x_2-1)\left[\frac{\alpha_2}{A^{x_2-x_1}}-\alpha_1  \right]=(2x_2-1)\left[\frac{a_2 \pi^{x_2-1}}{A^{x_2-x_1}\Gamma(x_2)}-\frac{a_1\pi^{x_2-1}}{\Gamma(x_1)}   \right]\geq 0
\end{align*}
and $\tilde{g}'_A$ is positive on $[1,+\infty)$. Thus $\tilde{g}_A$ is increasing on $[1,+\infty)$ and, always by assumption,
$$
g_A(1)=2\left(  \frac{\alpha_2}{A^{x_2-x_1}}-\alpha_1 \right)\geq 0
$$
hence $g_A(y)\geq 0$ on $[1,+\infty)$ and by Theorem \ref{THM1}, $\Lambda_A$ is the unique minimizer of $E_{V_{a,x}^{LJ}}$, up to rotation, among Bravais lattices of fixed area $A$.
\end{proof}

\begin{remark}
This bound is optimal for our method because, $g_A(1)=0$ for $\displaystyle A= \pi\left(\frac{a_2\Gamma(x_1)}{a_1\Gamma(x_2)} \right)^{\frac{1}{x_2-x_1}}$ and $A\mapsto g_A(1)$ is a decreasing function.
\end{remark}
\begin{example}
For $\displaystyle V(r)=\frac{1}{r^6}-\frac{2}{r^3}$ which corresponds to Lennard-Jones energy in our case in \cite{Betermin:2014fy}, we find
$$
\pi\left(\frac{a_2\Gamma(x_1)}{a_1\Gamma(x_2)} \right)^{\frac{1}{x_2-x_1}}=\pi\left( \frac{\Gamma(3)}{2\Gamma(6)} \right)^{1/3}=\frac{\pi}{120^{1/3}}.\\
$$
\end{example}

\noindent Now we prove that for small parameters, the global minimizer among all Bravais lattices - without area constraint - of our energy is unique and triangular. We follow some ideas from our previous paper \cite{Betermin:2014fy} which cannot be apply for classical Lennard-Jones potential $V_{LJ}(r)=r^{-12}-2r^{-6}$.\\

\begin{lemma}\label{UBArea} \textbf{(Upper bound for global minimizer's area)}
Let $L_{a,x}$ a global minimizer of $E_{V_{a,x}^{LJ}}$ among all Bravais lattices, then 
$$
\displaystyle |L_{a,x}|\leq \left(  \frac{a_2x_2}{a_1x_1}\right)^{\frac{1}{x_2-x_1}}.
$$
\end{lemma}
\begin{proof}
Same argument of STEP 3 in the proof of Theorem \ref{THM2}.B.1.\\
\end{proof}

\noindent Thus we can prove \textbf{Theorem \ref{THM2}.B.2} :
\begin{proof}
Let $L_{a,x}$ be a global minimizer of $E_{V_{a,x}^{LJ}}$. We have
$$
h(x_2)\leq h(x_1)\iff \pi\left( \frac{a_2\Gamma(x_1)}{a_1\Gamma(x_2)} \right)^{\frac{1}{x_2-x_1}}\geq \left(  \frac{a_2x_2}{a_1x_1}\right)^{\frac{1}{x_2-x_1}}
$$
 then by Lemma \ref{UBArea} we get
 $$
 |L_{a,x}|\leq \pi\left( \frac{a_2\Gamma(x_1)}{a_1\Gamma(x_2)} \right)^{\frac{1}{x_2-x_1}}
 $$
 and by Proposition \ref{LJones}, the minimizer among lattices of area $|L_{a,x}|$ fixed is unique and triangular, hence the global minimizer of the energy is unique and triangular. Furthermore, let 
$$
f(r):=E_{V_{a,x}^{LJ}}[r\Lambda_1]=a_2\zeta_{\Lambda_1}(2x_2)r^{-2x_2}-a_1\zeta_{\Lambda_1}(2x_1)r^{-2x_1}
$$
then we have $f'(r)=-2a_2x_2\zeta_{\Lambda_1}(2x_2)r^{-2x_2-1}+2a_1x_1\zeta_{\Lambda_1}(2x_1)r^{-2x_1-1}$ and
$$
f'(r)\geq 0 \iff r\geq \left(\frac{a_2x_2\zeta_{\Lambda_1}(2x_2)}{a_1x_1\zeta_{\Lambda_1}(2x_1)}  \right)^{\frac{1}{2(x_2-x_1)}}.
$$
Hence the minimizer of $E_{V_{a,x}^{LJ}}$ among triangular lattices is $\Lambda_{|L_{a,x}|}$ with 
$\displaystyle |L_{a,x}|=\left(\frac{a_2x_2\zeta_{\Lambda_1}(2x_2)}{a_1x_1\zeta_{\Lambda_1}(2x_1)}  \right)^{\frac{1}{x_2-x_1}}$.
\end{proof}

\begin{remark} For an easy numerical computation of  global minimizer's area, we can use formula \eqref{zetatriang} to obtain
$$
|L_{a,x}|=\frac{1}{2\sqrt{3}}\left( \frac{a_2x_2\zeta(x_2)(\zeta(x_2,1/3)-\zeta(x_2,2/3))}{a_1x_1\zeta(x_1)(\zeta(x_1,1/3)-\zeta(x_1,2/3))} \right)^{\frac{1}{x_2-x_1}}.
$$

\end{remark}

\begin{remark}
We can apply the previous Theorem to $x=(2,3)$, and $r\mapsto V_{a,x}^{LJ}(r^2)$ is a $(6-4)$ potential. Moreover we can choose $a_1$ and $a_2$ such that the well is as deep as we want.
\begin{center}
\includegraphics[width=10cm,height=80mm]{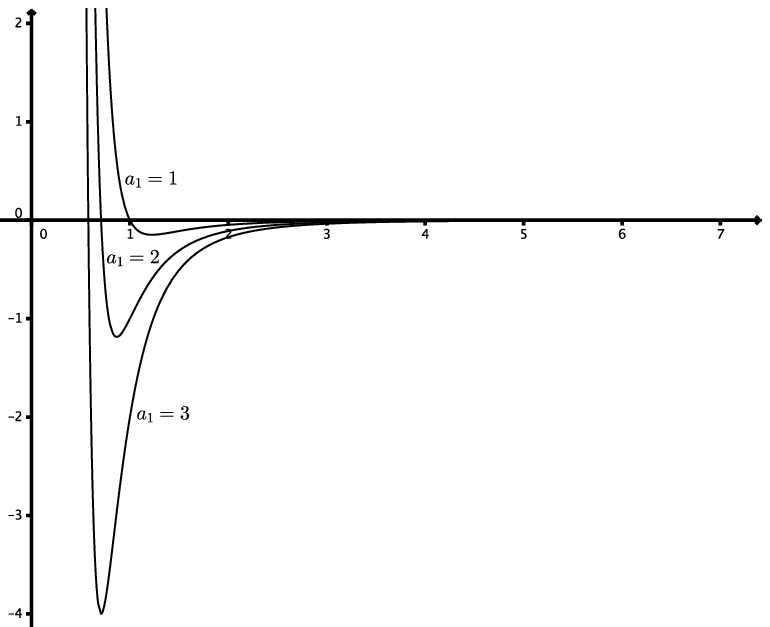}\\
\textbf{Fig. 7 :} Graphs of $\displaystyle r\mapsto \frac{1}{r^6}-\frac{a_1}{r^4}$ for $a_1\in \{1,2,3\}$.
\end{center}
\end{remark}
\noindent Now we explain a method to choose $x_1,x_2$ in order to have a triangular global minimizer and we give several numerical values.

\begin{lemma} \textbf{(Variations of h)}
Function $h$ is a decreasing function on $[1,\psi^{-1}(\log \pi)-1)$ and increasing on $[\psi^{-1}(\log \pi)-1,+\infty)$ where $\displaystyle \psi(x)=\frac{\Gamma'(x)}{\Gamma(x)}$ is the digamma function defined on $(0,+\infty)$.
\end{lemma}
\begin{proof}
We have $h'(t)=\Gamma(t)+t\Gamma(t)-t\log\pi\Gamma(t)$ and
$$
h'(t)\geq 0 \iff \psi(t)+\frac{1}{t}\geq \log\pi.
$$
We use the famous identity $\psi(t)+\frac{1}{t}=\psi(1+t)$ for any $t>0$ and we obtain, because $\psi$ is increasing on $(0,+\infty)$,
$$
h'(t)\geq 0 \iff t\geq \psi^{-1}(\log\pi)-1.
$$
\end{proof}
\begin{remark}
We compute $\psi^{-1}(\log\pi)-1\approx 2.6284732$ and we define $M\neq 1$ such that $h(M)=h(1)$. We have $M\approx 4.6022909$. Thus, if we want apply the previous theorem, it is clear that $x_1<\psi^{-1}(\log\pi)-1$ and $x_2<M$. Moreover, if we choose $x_1\in (1,\psi^{-1}(\log\pi)-1)$, we can choose $x_2\in(x_1,M_{x_1})$ where $M_{x_1}\neq x_1$ is such that $h(M_{x_1})=h(x_1)$.\\
\begin{center}
\includegraphics[width=12cm,height=80mm]{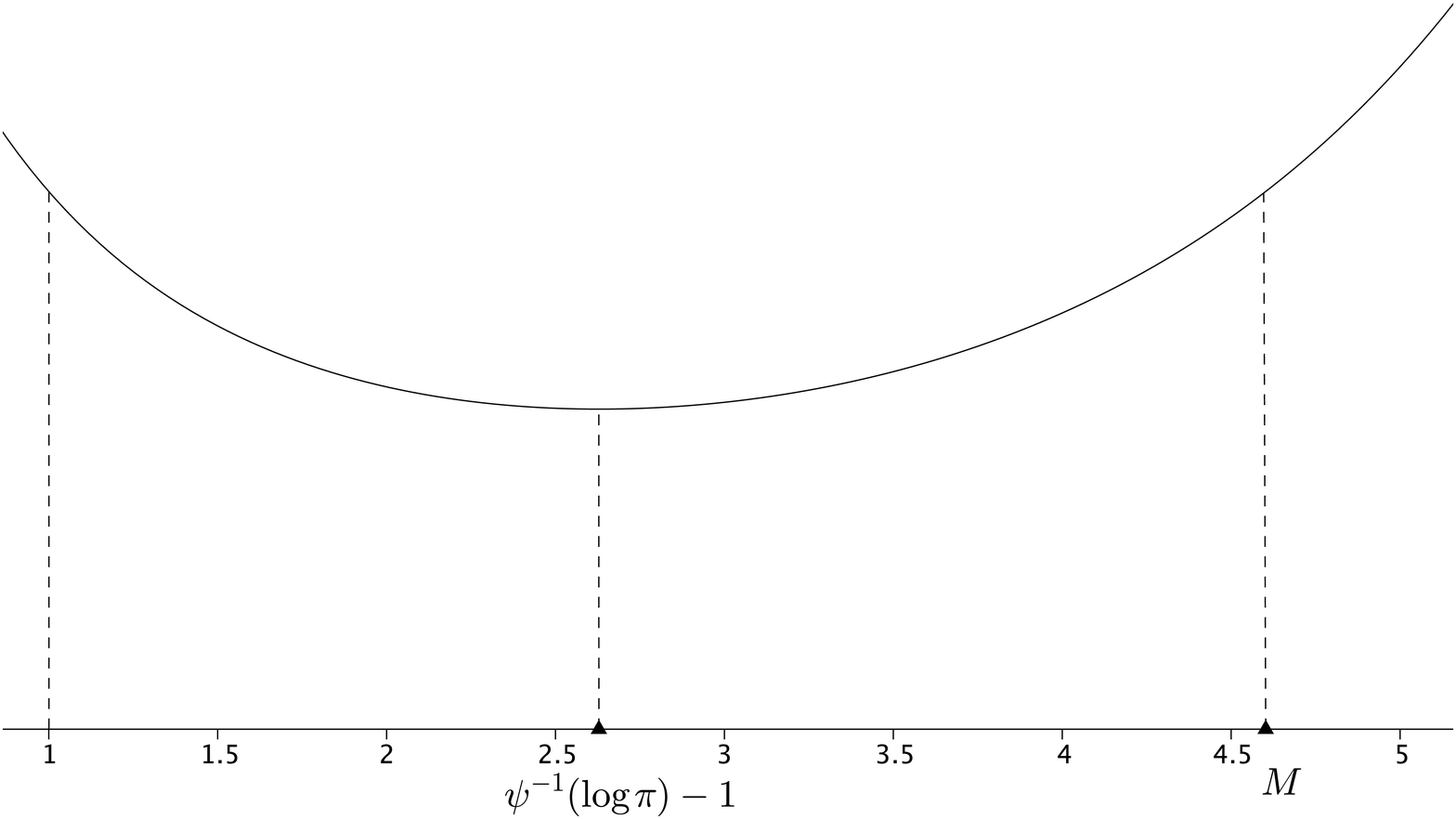}\\
\textbf{Fig. 8 :} Graph of $h$
\end{center}
Unfortunately we can only choose $x_2$ and $x_1$ such that $1<x_1<x_2<4.6022909$ and Lennard Jones case ($x_2=6$ and $x_1=3$) is not covered by our Theorem \ref{THM2}.B.2.\\ \\
We compute, for different values of $(x_1,x_2)$ satisfying $h(x_2)< h(x_1)$ and for $a=(1,1)$ :
\begin{itemize}
\item the value of the minimizer of $y\mapsto V_{a,x}^{LJ}(y^2)$, i.e. $\displaystyle y_{a,x}^{\min}:=\left( \frac{x_2}{x_1} \right)^{\frac{1}{2(x_2-x_1)}}$,
\item the value of the length of triangular global minimizer, i.e.  $L_{a,x}$, i.e. $r_{a,x}:=\sqrt{\frac{2|L_{a,x}|}{\sqrt{3}}}$,
\item the density\footnote{Here we exceptionally give values of densities -- and not areas -- more used in molecular simulations.} of $L_{a,x}$, i.e. $d_{a,x}:=|L_{a,x}|^{-1}$.
\end{itemize}
\begin{center}
\begin{tabular}[c]{|c| c| c| c| }
\hline
\backslashbox{$x_2$}{$x_1$} & $1.1$  & $1.5$  & $2$  \\
\hline $1.5$ & $(1.47,0.64,2.78)$ & \cellcolor{gray!40}  & \cellcolor{gray!40}  \\
\hline $2$ & $(1.39,0.80,1.82)$ & $(1.33,0.95,1.27)$ & \cellcolor{gray!40} \\
\hline $2.5$ & $(1.34,0.90,1.45)$ & $(1.29,1.02,1.10)$ & $(1.25,1.10,0.96)$ \\
\hline $3$ & $(1.30,0.95,1.27)$ & $(1.26,1.06,1.03)$ & $(1.22,1.11,0.93)$ \\
\hline $3.5$ & $(1.27,0.99,1.19)$ & $(1.24,1.08,1.00)$ & \cellcolor{gray!40}\\
\hline $4$ & $(1.25,1.01,1.14)$ & \cellcolor{gray!40}& \cellcolor{gray!40}\\
\hline
\end{tabular} \\
\vspace{2mm}
\textbf{Table 1. } Some values of $(y_{a,x}^{\min},r_{a,x},d_{a,x})$
\end{center}
\end{remark}

\noindent Obviously, we have non-optimality of $\Lambda_A$ if $A$ is sufficiently large, given by Proposition \ref{LDInvPower} :

\begin{prop}\label{UBound} \textbf{(Lennard-Jones at low density)}
Triangular lattice $\Lambda_A$ is a minimizer of $E_{V_{a,x}^{LJ}}$ among lattices of area $A$ fixed if and only if
$$
A\leq \inf_{|L|=1,L\neq \Lambda_1}\left( \frac{a_2(\zeta_L(2x_2)-\zeta_{\Lambda_1}(2x_2))}{a_1(\zeta_L(2x_1)-\zeta_{\Lambda_1}(2x_1))} \right)^{\frac{1}{x_2-x_1}},
$$
i.e. if $A$ is sufficiently large, $\Lambda_A$ is not a minimizer of $E_{V_{a,x}^{LJ}}$ among lattices of fixed area $A$.
\end{prop}
\begin{proof}
We apply directly Proposition \ref{LDInvPower} with $\sharp\{I_+\}=1$ and
$$
\inf_{L\neq \Lambda_1 \atop |L|=1} \max_{i\in I_+} \left\{\left( \frac{\sharp\{I_+\}a_i (\zeta_{L}(2x_i)-\zeta_{\Lambda_1}(2x_i))}{|a_1|(\zeta_{L}(2x_1)-\zeta_{\Lambda_1}(2x_1))} \right)^{\frac{1}{x_n-x_i}}  \right\}=\inf_{|L|=1,L\neq \Lambda_1}\left( \frac{a_2(\zeta_L(2x_2)-\zeta_{\Lambda_1}(2x_2))}{a_1(\zeta_L(2x_1)-\zeta_{\Lambda_1}(2x_1))} \right)^{\frac{1}{x_2-x_1}}.
$$
\end{proof}

\begin{remark}
More precisely we can found an explicit computable bound (but not optimal) if we take $L=\Z^2$ and use \eqref{zetasquare} and \eqref{zetatriang}. We give here densities $d_0$ such that for any $0<d<d_0$, $E_{V_{a,x}^{LJ}}[d^{-1/2}\Z^2]\leq E_{V_{a,x}^{LJ}}[d^{-1/2}\Lambda_1]$, i.e. square lattice have less energy than triangular lattice, with $a_1=a_2=1$.\\
\begin{center}
\begin{tabular}[c]{|c| c| c| c| c|c|c|c|c|c|c|}
\hline
\backslashbox{$x_2$}{$x_1$} & $1.1$ & $1.5$ & $2$ & $3$ & $4$ & $5$ & $6$ & $7$ & $8$ & $9$ \\
\hline $1.5$ & $0.05$ &\cellcolor{gray!40} & \cellcolor{gray!40}&\cellcolor{gray!40} & \cellcolor{gray!40}&\cellcolor{gray!40} &\cellcolor{gray!40} & \cellcolor{gray!40}& \cellcolor{gray!40}& \cellcolor{gray!40} \\
\hline $2$ & $0.14$ & $0.31$ &\cellcolor{gray!40} & \cellcolor{gray!40}&\cellcolor{gray!40} &\cellcolor{gray!40} &\cellcolor{gray!40} & \cellcolor{gray!40}&\cellcolor{gray!40} & \cellcolor{gray!40} \\
\hline $2.5$ & $0.21$ & $0.37$ & $0.43$ &\cellcolor{gray!40} & \cellcolor{gray!40}& \cellcolor{gray!40} &\cellcolor{gray!40} & \cellcolor{gray!40} & \cellcolor{gray!40}& \cellcolor{gray!40} \\
\hline $3$ & $0.27$ & $0.41$ & $0.47$ &\cellcolor{gray!40} &\cellcolor{gray!40} &\cellcolor{gray!40} &\cellcolor{gray!40} &\cellcolor{gray!40} & \cellcolor{gray!40}&\cellcolor{gray!40} \\
\hline $3.5$ & $0.31$ & $0.45$ & $0.50$ & $0.58$ &\cellcolor{gray!40} &\cellcolor{gray!40} &\cellcolor{gray!40} &\cellcolor{gray!40} &\cellcolor{gray!40} & \cellcolor{gray!40}\\
\hline $4$ & $0.35$ & $0.48$ & $0.53$ & $0.61$ &\cellcolor{gray!40} &\cellcolor{gray!40} &\cellcolor{gray!40} &\cellcolor{gray!40} &\cellcolor{gray!40} &\cellcolor{gray!40}  \\
\hline $5$ & $0.42$ & $0.53$ & $0.58$ & $0.65$ & $0.71$ &\cellcolor{gray!40} &\cellcolor{gray!40} &\cellcolor{gray!40} &\cellcolor{gray!40} &\cellcolor{gray!40}  \\
\hline $6$ & $0.47$ & $0.58$ & $0.63$ & $0.69$ & $0.74$ & $0.78$ &\cellcolor{gray!40} & \cellcolor{gray!40}& \cellcolor{gray!40}& \cellcolor{gray!40} \\
\hline $7$ & $0.52$ & $0.62$ & $0.66$ & $0.72$ & $0.77$ & $0.80$ &  $0.83$ &\cellcolor{gray!40} & \cellcolor{gray!40}&\cellcolor{gray!40}  \\
\hline $8$ & $0.56$ & $0.65$ & $0.69$ & $0.75$ & $0.79$ & $0.82$ & $0.84$ & $0.86$ & \cellcolor{gray!40}&\cellcolor{gray!40}  \\
\hline $9$ & $0.60$ & $0.68$ & $0.72$ & $0.77$ & $0.81$ & $0.84$ & $0.86$ & $0.88$ & $0.89$ & \cellcolor{gray!40}\\
\hline $10$ & $0.62$ & $0.70$ & $0.74$ & $0.79$ & $0.83$ & $0.85$ & $0.87$ & $0.89$ & $0.90$ & $0.91$ \\
\hline 
\end{tabular}\\ \vspace{2mm}
\textbf{Table 2.} Non-optimal critical densities for non-optimality of triangular lattice.
\end{center}
\end{remark}

\section{Potentials with exponential decay}

\subsection{ Definition and prove of Theorem \ref{THM1}.A for $f_{a,x,b,t}$}

\begin{defi}
Let $a=(a_1,...,a_n)\in (\R^*)^n$ with $a_n>0$, $x=(x_1,...,x_n)$ be such that $3/2<x_1<x_2<...<x_n$, $b=(b_1,...,b_m)\in (\R^*)^m$ and $t=(t_1,...,t_m)\in (\R^*)^m$, we define
$$
f_{a,x,b,t}(r):=\sum_{i=1}^n a_i r^{-x_i}+\sum_{j=1}^m b_j e^{-t_j\sqrt{r}}.
$$
We set $I_-:=\{i;a_i<0\}$ and $\displaystyle B:=\sum_{j=1}^m |b_j| t_j$.
\end{defi}

\begin{remark}
As explained in \cite{Koishi}, Fumi and Tosi proposed in \cite{FumiTosi} a potential for interaction between ions $Na^+$ and $Cl^-$ defined by
$$
V(r)=\frac{a_1}{r}+b_1e^{-t_1 r}-\frac{a_2}{r^6}-\frac{a_3}{r^8}.
$$
Obviously, potential $\displaystyle r\mapsto \frac{a_1}{r}$ is not admissible but the form of $V$ is close to $f_{a,x,b,t}$.
\end{remark}
\noindent Let us prove Theorem \ref{THM2}.A for this potential.

\begin{prop}\label{Genexpo} If we have
\begin{equation}\label{expgen}
A\leq \min \left\{\pi \min_{i\in I_-}\left\{ \left(\frac{a_n\Gamma(x_i)}{(2\sharp\{I_-\}+2)|a_i|\Gamma(x_n)}  \right)^{\frac{1}{x_n-x_i}} \right\}, \left( \frac{a_n \pi^{x_n+1}}{(\sharp\{ I_-\}+1)B \Gamma(x_n)} \right)^{\frac{1}{x_n+1/2}}  \right\}
\end{equation}
then $\Lambda_A$ is the unique minimizer of $E_{f_{a,x,b,t}}$, up to rotation, among Bravais lattices of fixed area $A$.
\end{prop}

\begin{proof}
As we have, by classical formula, for $a>0$,
$$
\mathcal{L}^{-1}[e^{-\sqrt{a.}}](y)=\frac{\sqrt{a}}{2\sqrt{\pi}}y^{-3/2}e^{-\displaystyle \frac{a}{4y}},
$$
taking $a=t_j^2$ for any $1\leq j\leq m$ and setting $\displaystyle \alpha_i=\frac{a_i\pi^{x_i-1}}{\Gamma(x_i)}$, we obtain, for any $y>0$,
\begin{align*}
\mu_{f_{a,x,b,t}}(y) &=\sum_{i=1}^n \alpha_i y^{x_i-1}+\sum_{j=1}^m \frac{b_j t_j}{2\sqrt{\pi}}y^{-3/2}e^{-\frac{t_j^2}{4y}}\geq \sum_{i=1}^n \alpha_i y^{x_i-1}-\frac{B}{2\sqrt{\pi}}y^{-3/2}
\end{align*}
and it follows that
\begin{align*}
g_A(y)&:=y^{-1}\mu_{f_{a,x,b,t}}\left(\frac{\pi}{yA}\right)+\mu_{f_{a,x,b,t}}\left(\frac{\pi y}{A}  \right)\\
&\geq \sum_{i=1}^n \frac{\alpha_i }{A^{x_i-1}}(y^{-x_i}+y^{x_i-1})-\frac{BA^{3/2}}{2\pi^2}\sqrt{y}-\frac{BA^{3/2}}{2\pi^2 y^{3/2}}\\
&=y^{-x_n}\left[\sum_{i=1}^n \frac{\alpha_i }{A^{x_i-1}}(y^{x_n-x_i}+y^{x_n+x_i-1})-\frac{BA^{3/2}}{2\pi^2}y^{x_n+1/2}-\frac{BA^{3/2}}{2\pi^2}y^{x_n-3/2}  \right].
\end{align*}
We set
$$
p_{a,x,b,t}(y):=\sum_{i=1}^n \frac{\alpha_i}{A^{x_i-1}}(y^{x_n-x_i}+y^{x_n+x_i-1})-\frac{BA^{3/2}}{2\pi^2}y^{x_n+1/2}-\frac{BA^{3/2}}{2\pi^2}y^{x_n-3/2} 
$$
and we notice that, for any $1\leq i\leq n$,
\begin{align*}
& x_n-x_i\neq x_n+1/2 \\
&x_n-x_i\neq x_n-3/2\\
&x_n+x_i-1 \neq x_n+1/2 \\
& x_n+x_i-1 \neq x_n-3/2
\end{align*}
because $x_i>3/2$. Hence the higher order term is $ \frac{\alpha_n}{A^{x_n-1}}y^{2x_n-1}$ with $\alpha_n>0$, and the number of negative terms is $2\sharp\{I_- \}+2$. Thus, by Cauchy's rule \eqref{UBoundRoot}, an upper bound on the values of the positive zero of $p_{a,x,b,t}$ is
\begin{align*}
M_{p_{a,x,b,t}}:=\max & \left\{\max_{i\in I_-}\left\{ \left( \frac{(2\sharp\{I_-\}+2)|\alpha_i|A^{x_n-x_i}}{\alpha_n} \right)^{\frac{1}{x_n+x_i-1}} \right\} ,\max_{i\in I_-}\left\{ \left( \frac{(2\sharp\{I_-\}+2)|\alpha_i|A^{x_n-x_i}}{\alpha_n} \right)^{\frac{1}{x_n-x_i}} \right\}, \right. \\
&\left. \left( \frac{B(2\sharp\{I_-\}+2)A^{x_n+1/2}}{2\pi^2\alpha_n} \right)^{\frac{1}{x_n-3/2}},\left( \frac{B(2\sharp\{I_-\}+2)A^{x_n+1/2}}{2\pi^2\alpha_n} \right)^{\frac{1}{x_n+1/2}}  \right\}.
\end{align*}
Now we have
\begin{align*}
& A\leq \min \left\{\pi \min_{i\in I_-}\left\{ \left(\frac{a_n\Gamma(x_i)}{(2\sharp\{I_-\}+2)|a_i|\Gamma(x_n)}  \right)^{\frac{1}{x_n-x_i}} \right\}, \left( \frac{2a_n \pi^{x_n+1}}{(2\sharp\{ I_-\}+2)B \Gamma(x_n)} \right)^{\frac{1}{x_n+1/2}}  \right\} \\
&\iff \forall i\in I_-, A\leq\pi\left(\frac{a_n\Gamma(x_i)}{(2\sharp\{I_-\}+2)|a_i|\Gamma(x_n)}  \right)^{\frac{1}{x_n-x_i}}\quad \text{ and } \quad A\leq \left( \frac{2a_n \pi^{x_n+1}}{(2\sharp\{ I_-\}+2)B \Gamma(x_n)} \right)^{\frac{1}{x_n+1/2}}\\
&\iff \forall i\in I_-, \frac{(2\sharp\{I_-\}+2)|\alpha_i|A^{x_n-x_i}}{\alpha_n}\leq 1 \quad \text{ and } \quad \frac{B(2\sharp\{I_-\}+2)A^{x_n+1/2}}{2\pi^2\alpha_n}\leq 1\\
&\iff M_{p_{a,x,b,t}}\leq 1,
\end{align*}
therefore if $y\geq 1\geq  M_{p_{a,x,b,t}}$ then $p_{a,x,b,t}(y)\geq 0$ hence $g_A(y)\geq 0$ and by Theorem \ref{THM1}, $\Lambda_A$ is the unique minimizer of $E_{f_{a,x,b,t}}$ among Bravais lattices of fixed area $A$.
\end{proof}

\begin{corollary}
If $I_-=\emptyset$ and
$$
A\leq \left( \frac{a_n \pi^{x_n+1}}{B\Gamma(x_n)} \right)^{\frac{1}{x_n+1/2}}
$$
then $\Lambda_A$ is the unique minimizer of $E_{f_{a,x,b,t}}$ among Bravais lattices of fixed area $A$.
\end{corollary}

\begin{remark}
Obviously, for any $A_0$, there exists $B$ sufficiently small such that for any $A\in(0,A_0]$, $\Lambda_{A_0}$ is the unique minimizer of our energy among Bravais lattices of fixed area $A$. We will study a simple particular case in next subsection in order to illustrate this fact. Furthermore we skipped the completely monotonic case but in the next following part we will give explicit condition for complete monotonicity in a simple case (see Proposition \ref{expCM}).
\end{remark}

\subsection{Example : opposite of Buckingham type potential}

In this part we study opposite of Buckingham type potential. Indeed, we cannot study Buckingham potential
$$
V_B(r)=a_1 e^{-\alpha r}-\frac{a_2}{r^6}-\frac{a_3}{r^8}
$$
because $\lim_{r\to 0} V_B(r)=-\infty$ and $\lim_{r\to +\infty} V_B(r)=0$ and it is sufficient to do $\|u\|\to 0$ in order to have $E_{V_B}[L]\to -\infty$. Hence we choose to treat simple general approximation of its opposite, well-adapted for our problem of minimization among Bravais lattices. Moreover we simplify notations in order to have only two parameters :

\begin{defi} For $a=(a_1,a_2)\in (0,+\infty)^2$ and for $x=(x_1,x_2)\in (0,+\infty)\times (3/2,+\infty)$, we define, for $r>0$, 
$$
f_{a,x}(r)=a_2 r^{-x_2}-a_1e^{-x_1\sqrt{r}}.
$$
\end{defi}

\begin{lemma} \textbf{(Variations of potential $r\mapsto f_{a,x}(r^2)$)}
We have the following two cases :
\begin{enumerate}
\item if $\displaystyle (2x_2+1)\left[\ln\left( \frac{2x_2+1}{x_1} \right)-1\right]\leq \ln\left(  \frac{2a_2x_2}{a_1x_1}\right)$, then $r\mapsto f_{a,x}(r^2)$ is decreasing on $(0,+\infty)$;
\item if $\displaystyle (2x_2+1)\left[\ln\left( \frac{2x_2+1}{x_1} \right)-1\right]> \ln\left(  \frac{2a_2x_2}{a_1x_1}\right)$ then there exists $r_m,r_M \in (0,+\infty)$ such that $r_m<\frac{2x_2+1}{x_1}<r_M$ and $r\mapsto f_{a,x}(r^2)$ is decreasing on intervals $(0,r_m)$ and $(r_M,+\infty)$ and increasing on $(r_m,r_M)$.
\end{enumerate}
\end{lemma}
\begin{proof}
We have $f(r):=f_{a,x}(r^2)=a_2 r^{-2x_2}-a_1e^{-x_1 r}$ and
$$
f'(r)=-\frac{2a_2x_2}{r^{2x_2+1}}+a_1x_1e^{-x_1r}.
$$
Thus we get
$$
f'(r)\geq 0 \iff e^{-x_1r}r^{2x_2+1}\geq \frac{2a_2x_2}{a_1x_1}\iff g(r)\geq 0
$$
where
$$
g(r)=-x_1r+(2x_2+1)\ln r -\ln\left(  \frac{2a_2x_2}{a_1x_1}\right).
$$
As $\displaystyle g'(r)=\frac{-x_1r+2x_2+1}{r}$, $g$ is  increasing on $\left(0,\frac{2x_2+1}{x_1}\right)$ and decreasing on $\left(\frac{2x_2+1}{x_1},+\infty  \right)$. Moreover $g(r)$ goes to $-\infty$ as $r\to 0$ or $r\to +\infty$.\\ Hence if $g\left(\frac{2x_2+1}{x_1}  \right)\leq 0$, i.e.
$$
(2x_2+1)\left[\ln\left( \frac{2x_2+1}{x_1} \right)-1\right]\leq \ln\left(  \frac{2a_2x_2}{a_1x_1}\right)
$$
then $g(r)\leq 0$ and $f'(r)\leq 0$ on $(0,+\infty)$, i.e. $f$ is decreasing on $(0,+\infty)$.\\
Furthermore, if $g\left(\frac{2x_2+1}{x_1}  \right)> 0$ then there exists $r_m,r_M$ such that $r_m<\frac{2x_2+1}{x_1}<r_M$ and $f$ is decreasing on intervals $(0,r_m)$ and $(r_M,+\infty)$ and increasing on $(r_m,r_M)$. \\ \\
\end{proof}

\begin{center}
\includegraphics[width=8cm,height=60mm]{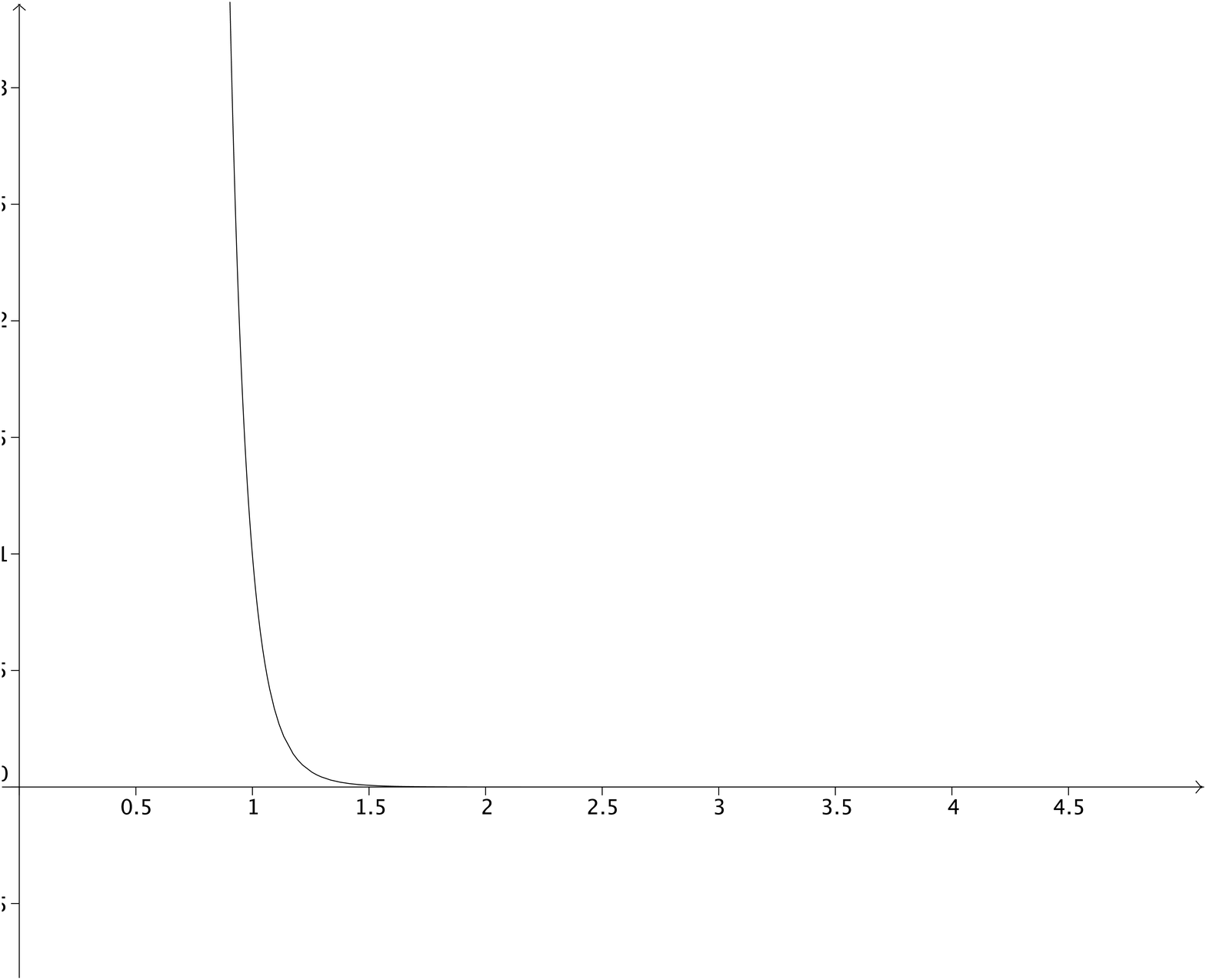} \includegraphics[width=8cm,height=60mm]{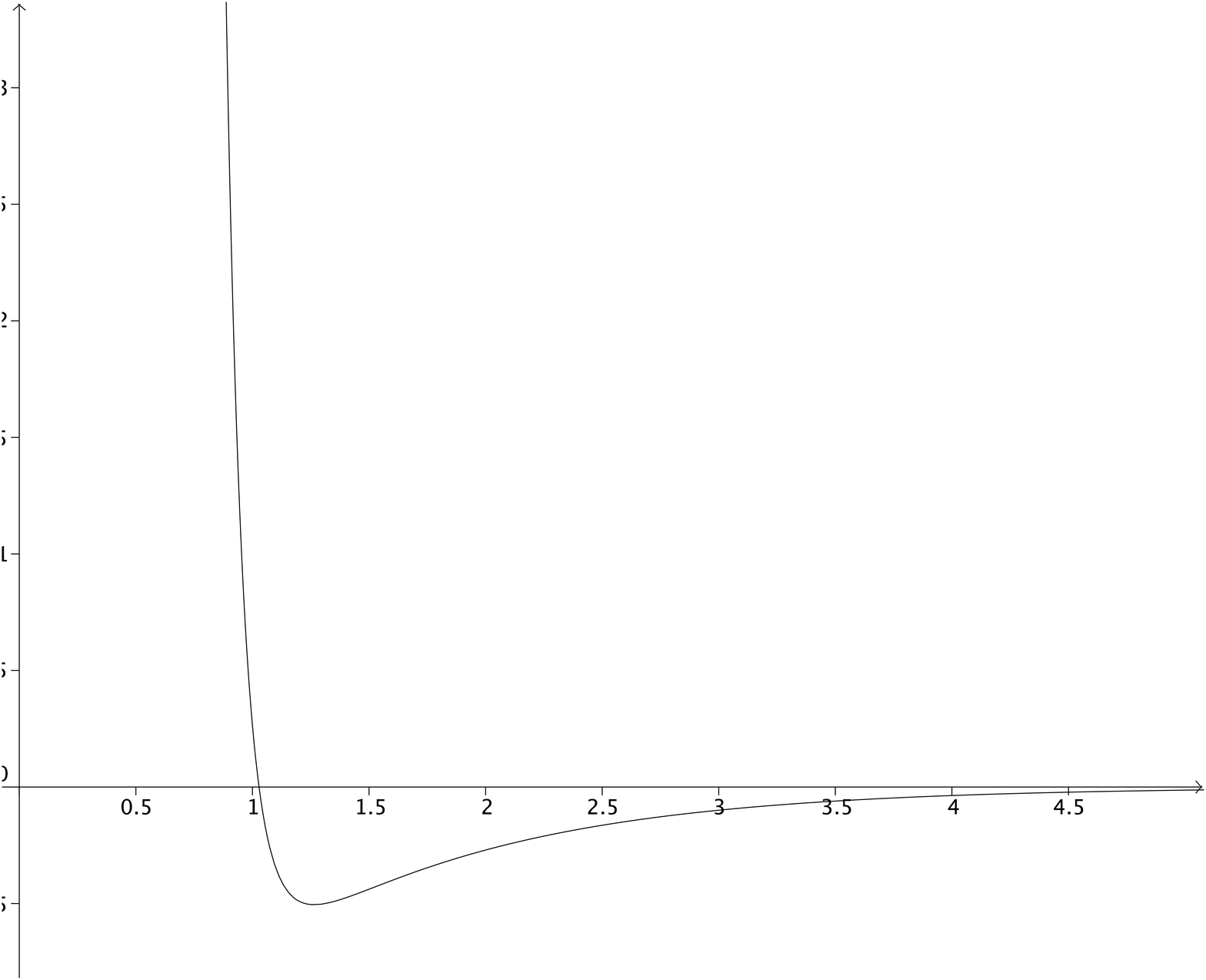}\\
\textbf{Fig. 9 :} Graph of $r\mapsto f_{a,x}(r^2)$ for $a=(1,1)$, $x=(5,6)$ on the left and $a=(1,2)$, $x=(1,6)$ on the right.
\end{center}
\vspace{5mm}

\begin{prop}\label{expCM} We have the following two cases :
\begin{itemize}
\item If it holds
$$
(x_2+1/2)\left[1+\ln\left( \frac{x_1^2}{4x_2+2} \right)  \right]\geq \ln\left(\frac{a_1x_1\Gamma(x_2)}{2\sqrt{\pi}a_2}  \right)
$$
then for any $A>0$, $\Lambda_A$ is the unique minimizer of $E_{f_{a,x}}$, up to rotation, among Bravais lattices of fixed area $A$.
\item If it holds
$$
A\leq \left(\frac{a_2\pi^{x_2+1}}{a_1x_1\Gamma(x_2)}  \right)^{\frac{1}{x_2+1/2}} \quad \text{ and } \quad (x_2+1/2)\left[1+\ln\left( \frac{x_1^2}{4x_2+2} \right)  \right]< \ln\left(\frac{a_1x_1\Gamma(x_2)}{2\sqrt{\pi}a_2}  \right)
$$
then $\Lambda_A$ is the unique minimizer of $E_{f_{a,x}}$, up to rotation, among Bravais lattices with fixed area $A$. Moreover, for any $a\in (0,+\infty)^2$, $x_2>3/2$, $A_0>0$ and any $x_1$ such that
$$
0<x_1\leq C_{A_0}:=\frac{a_2 \pi^{x_2+1}}{a_1 A_0^{x_2+1/2}\Gamma(x_2)},
$$
$\Lambda_A$ is the unique minimizer of $E_{f_{a,x}}$, up to rotation, among Bravais lattices of fixed area $A\in (0,A_0]$.
\end{itemize}
\end{prop}

\begin{proof}
By classical formula, we get
$$
\mu_{f_{a,x}}(y)=\frac{a_2}{\Gamma(x_2)}y^{x_2-1}-\frac{a_1x_1}{2\sqrt{\pi}}y^{-3/2}e^{-\displaystyle \frac{x_1^2}{4y}}.
$$
Our theorem is a consequence of Proposition \ref{CrystCM} because
$$
\forall y>0, \mu_{f_{a,x}}(y)\geq 0 \iff (x_2+1/2)\left[1+\ln\left( \frac{x_1^2}{4x_2+2} \right)  \right]\geq \ln\left(\frac{a_1x_1\Gamma(x_2)}{2\sqrt{\pi}a_2}  \right).
$$
Indeed, we have
\begin{align*}
\forall y>0, \mu_{f_{a,x}}(y)\geq 0 &\iff e^{x_1^2/4y}y^{x_2+1/2}\geq \frac{a_1 x_1 \Gamma(x_2)}{2\sqrt{\pi}a_2}\\
&\iff \frac{x_1^2}{4y}+(x_2+1/2)\ln y - \ln\left( \frac{a_1 x_1 \Gamma(x_2)}{2\sqrt{\pi}a_2}\right)\geq 0.
\end{align*}
We set 
$$
g(y)=\frac{x_1^2}{4y}+(x_2+1/2)\ln y - \ln\left( \frac{a_1 x_1 \Gamma(x_2)}{2\sqrt{\pi}a_2}\right)
$$
and we have $g'(y)=-\frac{x_1^2}{4y^2}+\frac{x_2+1/2}{y}$. It follows that $g$ is decreasing on $\left(0, \frac{x_1^2}{4x_2+2}\right)$ and increasing on $\left(\frac{x_1^2}{4x_2+2}, +\infty  \right)$. As $g$ goes to $+\infty$ as $y$ goes to $0$ or $+\infty$, it is clear that
\begin{align*}
\forall y>0, g(y)\geq 0 &\iff g\left(\frac{x_1^2}{4x_2+2}  \right)\geq 0\\
&\iff (x_2+1/2)\left[1+\ln\left( \frac{x_1^2}{4x_2+2} \right)  \right]\geq \ln\left(\frac{a_1x_1\Gamma(x_2)}{2\sqrt{\pi}a_2}  \right).
\end{align*}
Now, if $f_{a,x}$ is not completely monotonic, we apply directly Proposition \ref{Genexpo} to obtain second point.\\
Third point is clear because for any $(a_1,a_2)\in (0,+\infty)^2$ and any $x_2>3/2$,
$$
x_1\mapsto \left(\frac{a_2\pi^{x_2+1}}{a_1x_1\Gamma(x_2)}  \right)^{\frac{1}{x_2+1/2}}
$$ 
is an increasing function which goes to infinity as $x_1\to 0$.
\end{proof}

\begin{example}
For instance, we can choose $a=(1,1)$, $x_2=6$ and $A_0=1$. Thus we get 
$$
C_1=\frac{\pi^{13}}{11!}\approx 0.0727432
$$
and for any $x_1\leq C_1$, $\Lambda_1$ is the unique minimizer of $E_{f_{a,x}}$ among Bravais lattices of unit fixed area.\\
Moreover the form of the potential $y\mapsto f_{a,x}(y^2)$ is such that the decay to $0$ at infinity is slow as $x_1$ goes to 0.
\begin{center}
\includegraphics[width=12cm,height=90mm]{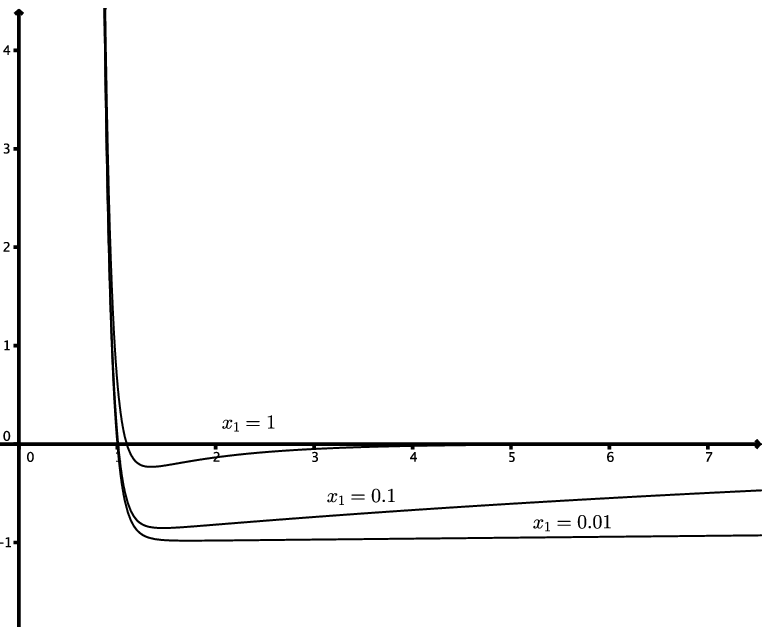}\\
\textbf{Fig. 10 :} Graph of $\displaystyle y\mapsto f_{a,x}(y^2)=\frac{1}{y^{12}}-e^{-x_1y}$ for $x_1\in \{0.01,0.1,1 \}$.
\end{center}
\end{example}
\begin{remark}
Our argument used in proofs of Theorem \ref{THM2}, based on variations of potential, can't be applied for our potentials $f_{a,x}$. 
\end{remark}

\noindent \textbf{Acknowledgements:} I am grateful to Etienne Sandier, Florian Theil, Salvatore Torquato and my colleague Peng Zhang for their interest and helpful discussions.

\bibliographystyle{plain}
\bibliography{biblio}

\end{document}